\documentclass[acmsmall]{acmart}
\pdfoutput=1

\DeclareUnicodeCharacter{2208}{\ensuremath{\in}}

\usepackage{xr}

\usepackage{mathpartir}
\usepackage{bcprules}
\usepackage{physics}
\usepackage{enumitem}

\usepackage{amsmath}
\usepackage{mathdots}
\usepackage{cancel}
\usepackage{color}
\usepackage{siunitx}
\usepackage{array}
\usepackage{multirow}
\usepackage{gensymb}
\usepackage{tabularx}
\usepackage{extarrows}
\usepackage{booktabs}
\usepackage{alltt}
\usepackage[capitalize,nameinlink,noabbrev]{cleveref} 

\newtheorem{theorem}{Theorem}[section]
\newtheorem{lemma}[theorem]{Lemma}
\newtheorem{definition}[theorem]{Definition}

\newcommand{\app}{\,\texttt{++}\,}

\newcommand{\length}{\texttt{length}\,}
\newcommand{\DownarrowA}{\Downarrow_\mathcal{A}}

\newcommand{\DownarrowC}{\Downarrow_\mathcal{C}}

\newcommand{\NdetDownarrowA}{\Downarrow_\mathcal{A}^\star}
\newcommand{\leak}{\mathtt{Leak}\,}
\newcommand{\leaknoarg}{\mathtt{Leak}}
\newcommand{\branch}{\mathtt{CompNonDet}\,}
\newcommand{\branchnoarg}{\mathtt{CompNonDet}}
\newcommand{\treeleak}[2]{\mathtt{TreeLeak}\,#1\,#2}
\newcommand{\treeleaf}{\mathtt{TreeLeaf}}
\newcommand{\treebranch}[1]{\mathtt{TreeBranch}\,#1}
\newcommand{\bigtreeleak}[2]{\mathtt{BigTreeLeak}\,#1\,#2}
\newcommand{\bigtreeleaf}{\mathtt{BigTreeLeaf}}
\newcommand{\bigtreebranch}[1]{\mathtt{BigTreeBranch}\,#1}
\newcommand{\en}{\mathtt{End}}

\begin{document}

\title{Smooth, Integrated Proofs of Cryptographic Constant Time for Nondeterministic Programs and Compilers}

\author{Owen Conoly}
\orcid{0009-0003-3129-1218}
\affiliation{%
  \institution{Massachusetts Institute of Technology}
  \city{Cambridge}
  \country{USA}
}
\email{owenc@mit.edu}

\author{Andres Erbsen}
\orcid{0000-0002-9854-7500}
\affiliation{%
  \institution{Google}
  \city{Cambridge}
  \country{USA}
}
\email{andreser@mit.edu}

\author{Adam Chlipala}
\orcid{0000-0001-7085-9417}
\affiliation{%
  \institution{Massachusetts Institute of Technology}
  \city{Cambridge}
  \country{USA}
}
\email{adamc@csail.mit.edu}

\setcopyright{cc}
\setcctype{by}
\acmDOI{10.1145/3729318}
\acmYear{2025}
\acmJournal{PACMPL}
\acmVolume{9}
\acmNumber{PLDI}
\acmArticle{215}
\acmMonth{6}
\received{2024-11-14}
\received[accepted]{2025-03-06}


\begin{abstract}
Formal verification of software and compilers has been used to rule out large classes of security-critical issues, but risk of unintentional information leakage has received much less consideration.
It is a key requirement for formal specifications to leave some details of a system's behavior unspecified so that future implementation changes can be accommodated, and yet it is nonetheless expected that these choices would not be made based on confidential information the system handles.
This paper formalizes that notion using omnisemantics and plain single-copy assertions, giving for the first time a specification of what it means for a nondeterministic program to be constant-time or more generally to avoid leaking (a part of) its inputs.
We use this theory to prove data-leak-free execution of core cryptographic routines compiled from Bedrock2 C to RISC-V machine code, showing that the smooth specification and proof experience omnisemantics provides for nondeterminism extends to constant-time properties in the same setting.
We also study variants of the key program-compiler contract, highlighting pitfalls of tempting simplifications and subtle consequences of how inputs to nondeterministic choices are constrained.
Our results are backed by modular program-logic and compiler-correctness theorems, and they integrate into a neat end-to-end theorem in the Coq proof assistant.
\end{abstract}

\ccsdesc[500]{Software and its engineering~Formal software verification}
\ccsdesc[500]{Software and its engineering~Compilers}
\ccsdesc[500]{Security and privacy~Software and application security}

\keywords{cryptographic constant time, compiler verification, omnisemantics}

\maketitle

\section{Introduction}

Cryptographic software relies on compiler correctness beyond what is captured by a conventional trace-inclusion or specification-preservation result.
The crux is that while a number of implementation choices are intentionally left to the compiler, it would not be fair for the compiler to generate code that makes these choices based on examination of confidential data handled by the program.
This consideration applies both to data flows that are commonly modeled in compiler-correctness theorems---for example, when iteration order of collections depends on addresses of stack-allocated objects---and also to ``side'' channels that are often omitted, for example utilization of CPU resources.
In either case, it would be unreasonably burdensome for the source-language semantics to specify these details fully, and actually preserving them in the compiled code would severely limit implementation flexibility, completely prohibiting optimization if this idea is taken literally.

Past study~\cite{CtCompCert} of confidentiality-preserving compilation in the context of CompCert~\cite{Leroy-backend} leaves these issues out-of-scope: ``we assume that the languages are deterministic,'' ``it does not seem straightforward what it would mean for a non-deterministic C program to be constant-time.''
Impressive progress (porting of 17 out of 20 compiler passes) was possible in spite of leaving these fundamental questions unanswered because CompCert's proof strategy already relies on an intricate deterministic memory model.

We present our development of a new generalization of cryptographic constant time to support nondeterminism.
We implement and evaluate this generalization through a full verified software stack in Coq, including program-logic proofs of C-like programs and verified compilation to machine code, with derived timing-independence theorems at the machine-code level.
What is most distinctive about our approach is recognition that rules about which program inputs may influence which nondeterministic choices are inherently program-specific, and so we need an expressive specification formalism that pairs well with a compiler proved to preserve specifications.
We study several candidate formalisms, summarize their pros and cons, and relate them formally.

\subsection{Constant-Time Programming}
Let us begin by reviewing the established programming discipline known as \emph{cryptographic constant time}. 
First, despite the name, this discipline enforces a data-flow property, not a statement about measured or asymptotic execution time. 
The vast majority of important cryptographic implementations follow this discipline to avoid leaking secrets (e.g. secret keys or plaintexts) through side channels such as timing of program execution, cache occupancy, or CPU resource usage.
The same discipline can also be captured using straightforward operational rules specifying information leaked by execution steps.
The standard approach that we follow is to apply an instrumented semantics that, as a program runs, accumulates a \emph{leakage trace} of the following kinds of events: read memory at address $a$, wrote memory at address $a$, or resolved a conditional jump (given as a test expression's Boolean value).
Note that we log memory \emph{addresses} but not the values read/written, and that the branch decisions being leaked allows control flow to be reconstructed from the program and its leakage trace.
If a deterministic program's input is partitioned into secret and public components, the program is said to be constant-time if, on any two executions \emph{that agree on values of all public inputs}, those two executions \emph{generate identical leakage traces}.
It is apparent that constant time generalizes noninterference, where we consider leakage traces as implicit public outputs.

This condition composes well with the vast majority of important compiler and hardware optimizations, with optimizations that are exceptions to this pattern becoming known as microarchitectural side-channel vulnerabilities.
For instance, all common systems of memory caching retain secret-independent timing with constant-time programs if configured appropriately, since their ``control flow'' (e.g., interactions between caches in a cache-coherence protocol) only changes based on memory addresses accessed, not values read/written.
However, it is apparent that the definition we just presented no longer captures the desired concept in the presence of nondeterminism.
Calling a function twice on \emph{the same} arguments may yield different leakage traces, so we have little hope that calls with different arguments will always generate identical traces.
Yet programs expected to be side channel-secure often \emph{do} depend on nondeterminism, whether as a consequence of engineering abstraction (e.g. memory-management details are left unspecified at the source level) or inherently through I/O (e.g. to implement cryptographic protocols, generate randomness, or access long-term keys).
Key-and-input-dependent variation in response (output) time has been exploited to perform padding-oracle attacks even against systems that do not deliberately disclose information about unauthenticated decryption results~\cite{Lucky13}, and constant-time techniques to avoid this issue rely on stack-allocated temporary buffers to store sensitive intermediate values.\footnote{For example, see BoringSSL function \texttt{EVP\_tls\_cbc\_copy\_mac}.}

\subsection{Our Semantics Approach in One Page}\label{semapproach}

All source-language and machine-code programs in this paper will be analyzed in terms of which specifications they admit in the sense of total-correctness Hoare logic or omnisemantics~\cite{Omnisemantics} (a variation of weakest-precondition calculus that we explain more as we reference aspects of it later).
In symbols, $p \Downarrow Q$ means that program $p$ (which we in this section treat as including all relevant initial state) terminates in a state satisfying the predicate $Q$, and it does so \emph{regardless of the outcomes of nondeterministic choices} during its execution.
Note that the set $Q$ is just a constraint on the possible behaviors, and there is no requirement that every outcome in $Q$ is even possible, let alone as likely to occur as others.
Thus, we allow compilers to shrink $Q$, perhaps by compiling a nondeterministic source program to a deterministic machine-code program.

Side-channel leakage is encoded as yet another component of the state of a program, and $Q$ is asserted on it upon termination.
For specific programs $p$, we pick concrete specifications $Q$ that relate the runtime input-output log $t$ of each execution to its leakage trace $k$.
For example, saying that a username $u$ may be leaked but a passphrase $s$ may not: $\exists f. \; p \Downarrow \{ (t, k) \; | \; \exists u, s.\; t = [\verb|IN | u; \verb|IN | s] \wedge k = f(u) \}$.
Here the I/O log $t$ is constrained to record inputting just $u$ and $s$, in that order.
Critically, quantifier ordering requires the program to satisfy this specification using one and the same leakage-prediction function $f$ for all passphrases $s$.
If $f$ were existentially quantified after $s$, it could be the constant function that returns $s$, allowing the passphrase to be leaked.

The example just presented shows how to encode a leakage specification as a ``single-copy'' property: a postcondition that all possible executions must satisfy individually.
The ability to capture properties that might also be stated in a ``two-copy,'' hyperproperty style (e.g., saying that for every pair of possible executions, their leakage traces are related) arises from the ability to share quantifiers between the postconditions of the possible executions: we can say that there must exist a reference leakage trace that every possible leakage trace is related to.
For example, the example above says that every leakage trace is related to the same function \(f\).
Encoding leakage specifications as single-copy rather than two-copy properties greatly simplifies our formalization.

We formalize compiler correctness in the style of specification preservation.
Specifically, an optimization pass $C$ that does not change the language or the possible leakage must satisfy $\forall p, Q. \; p\Downarrow Q \implies C(p)\Downarrow Q$.
Compiler passes that change the language or the program state (including leakage) would use different variants of $\Downarrow$ and $Q$ in the conclusion.
A common pattern is to define the lower-level postcondition $Q'$ in terms of a state-representation relation $R$ and the high-level postcondition $Q$.
For example, $Q'$ may assert that there exists a machine word that is stored in a calling convention-specified register and also is a permissible return value of the function according to the high-level postcondition $Q$.
Similarly, the leakage specification of a simple compiler pass can state that every leakage trace of the compiled program equals the output of some (global) function $g$ on some leakage trace of the source program:
$\forall p, Q.\; \exists g.\; p\Downarrow Q \implies C(p)\Downarrow
\{ (t, g(k)) \mid Q(t, k) \}$.

The situation is more complicated when program control flow can depend on nondeterministic choices that may be resolved by the compiler, perhaps in a way that depends on some part of the program's state but not another.
We will now describe scenarios where these considerations come up and develop ways to specify how we would want a compiler to handle them.

The end point of our journey will be a combined theory including:
\begin{itemize}
\item The \textbf{first generalization of cryptographic constant time to support nondeterminism},
\item that can specify \textbf{which nondeterministic choices may be resolved based on secrets},
\item which we then use for \textbf{verification of multipass compilers that maintain constant time as they reduce some but not all nondeterminism} as we progress through compilation,
\item and to give \textbf{compiler proof-compatible Hoare-logic certificates to specific C programs representative of constant time-programming challenges},
\item all evaluated through \textbf{a Coq development}.
\end{itemize}

\subsection{Overview of the Paper}
\cref{sec:det} describes the basic idea of our approach: to formalize ``program \(p\) is constant-time'' as a single-copy property, just say ``the leakage trace of \(p\) is a function of public inputs.''
\cref{stack_swap_bad} notes that this approach runs into problems when what we call \emph{compiler-resolved nondeterminism} is involved, but \cref{fix_stack_swap} presents a simple solution: get rid of the problematic nondeterminism by parameterizing the source-level semantics over an \emph{oracle}.
\par It is natural to ask whether our approach can be made to work without getting rid of the source-level nondeterminism.
In \cref{sec:two}, we show that the answer is yes: we provide an alternative approach without determinizing the source-level semantics, and then we prove it equivalent (\cref{equiv}) to the semantics-determinization approach.
In \cref{predictors}, we show again (in a different, subtly nonequivalent way) that the answer is yes, this time using the concept of a \emph{predictor}, which is a dual concept to the oracles used in \cref{sec:determinize} and \cref{sec:two}.
\par In \cref{sec:compiler}, we shift our focus to compiler specifications.
In \cref{leakage_trans_defn}, we propose a simple definition of ``constant time-preserving'' for a compiler: the compiler should admit \emph{leakage-transformation functions}.
\cref{ltf_not_enough} observes that requiring a compiler to admit leakage-transformation functions is insufficient when compiler-resolved nondeterminism is involved.
As a fix, we propose two stronger compiler specifications: \cref{ctp5} basically says that constant-time specifications in the style of \cref{sec:determinize} (or, equivalently, \cref{sec:two}) are preserved by the compiler, while \cref{ctp3} basically says that constant-time specifications in the style of \cref{predictors} are preserved.
It turns out that \cref{ctp5} and \cref{ctp3} are inequivalent in interesting ways; \cref{sec:compiler} concludes with a discussion of the tradeoffs between them.
\par We conclude with case studies.
\cref{compiler_proof} describes our strategies for proving that the Bedrock2 compiler is constant time-preserving, and \cref{case_studies} describes our experience proving some Bedrock2 source programs correct and composing source-program theorems with the compiler theorem.

In total, the paper presents three different approaches to verifying constant time for source programs (and three corresponding compiler specifications): semantics determinization via oracles (\cref{fix_stack_swap}); an equivalent approach, still with oracles but without semantics determinization (\cref{sec:two}); and the approach with predictors (\cref{predictors}).

Our implementation is available as open source: compiler and source-program proofs in the semantics-determinization style appear in the main repository for Bedrock2\footnote{\url{https://github.com/mit-plv/bedrock2}}.
Proofs of the Bedrock2 compiler in all three styles are implemented in the artifact associated with this paper.
Also, every fact presented in this paper as a theorem, lemma, or corollary is formalized in a separate repository\footnote{\url{https://github.com/OwenConoly/semantics_relations}}.

\section{The Languages Used in This Paper}
  
\subsection{Background on Bedrock2 (Our Source Language)}\label{sec:background}

Bedrock2 is an imperative language inspired by K\&R~C~\cite{KRC88}.
It has been used to create \emph{end-to-end Coq proofs} about: a complete software-hardware IoT device controlling a lightbulb based on network inputs~\cite{lightbulb}, a cryptographic server running on a microcontroller~\cite{garagedoor}, and a selection of data structures~\cite{live}.
As Bedrock2 has previously been used for ``one Q.E.D.'' proofs that span the abstraction gap from relational specifications of applications down to hardware designs that run on FPGAs, we considered it a compelling baseline to extend in this work, to show that our new style of reasoning about timing channels is likely to apply to deep verified stacks, as well.

Syntax includes expressions and commands, with all variables storing machine words.
Imperative state includes locals values $\ell$ and byte-granularity memory $m$.
The state of the program also contains an input-output log $t$, preserved by compilation; and our addition, the leakage trace $k$.

Bedrock2 language features include functions, loops, conditionals, stack allocation, and I/O.
There are no restrictions on pointer operations; pointers are simply machine words indexing into one flat memory space.
They can be added, compared, printed, etc.
Pointer-based data structures (e.g. linked lists) can be implemented in Bedrock2 as in C.
Furthermore, Bedrock2's lower-level semantics simplify verification of memory-management code, as demonstrated by past work~\cite{live}.

In this extension of Bedrock2 to handle constant time, we adapt the source-language semantics and compiler proof, but we leave for future work adaptation of processor proofs.
As a result, our final theorems are about RISC-V machine code instead of hardware.
However, machine-code semantics is our only trusted code; we do not trust the Bedrock2 source language or compiler.

It is also worth brief description of other complexities of the Bedrock2 compiler that we inherit.
It works with not just the source language and RISC-V machine language but also one intermediate language flattened into three-address code.
There are six phases: flattening, compilation of immediates, dead-code elimination, register allocation, spilling, and machine-code generation.

Two phases of the Bedrock2 compiler have different source and target languages.
In our experience, verification of security properties has little interesting interaction with the difficulties of language heterogeneity; our techniques extend to the heterogeneous case without complication.
So, for simplicity, our discussion will mostly stick to phases with matching source and target languages.

\subsection{Bedrock2 Semantics}\label{bedrock2_sem}
The canonical definition of Bedrock2 semantics is given as a set of weakest-precondition predicate-transformer rules that form an inductive definition of $p\Downarrow{}Q$.
This form is closely related to traditional big-step semantics (for a well-defined terminating program, $\bigcap\left\{Q : p\Downarrow Q\right\}$ is the set of possible outcomes) and suitable for compiler verification by forward specification preservation~\cite{Omnisemantics}.
\par Our semantics are the same as in the Bedrock2 paper~\cite{lightbulb}, except we have added leakage traces \(k\).
A leakage trace includes two types of events: \emph{leakage} and \emph{compile-time nondeterminism} events.
The $\branchnoarg$ events record how \emph{compiler-resolved} nondeterministic events were resolved, while the $\mathtt{Leak}$ events, as in the previous high-level description of leakage trace, represent information that is leaked to an adversary.
Concretely, we write a leakage trace as a list---for example,
\[[\leak x_1, \branch y_1, \leak x_2, \leak x_3, \branch y_2, \branch y_3, ...],\]
where the \(x_i\) and \(y_i\) are machine words.
In Bedrock2, $\branchnoarg$ events appear only in the stack-allocation rule (\cref{fix_stack_swap}).
However, our techniques apply equally well to other sources of compiler-resolved nondeterminism: nondeterministic evaluation order for expressions, PRNGs, etc.
\par The following big-step-omnisemantics~\cite[\S 2]{Omnisemantics} rules illustrate our Bedrock2 semantics \(\Downarrow\).
\begin{mathpar}
\inferrule[eval-store]
{(x,a) \in \ell \\ (a+n) \in \mathtt{dom}\;m \\\\
(y,v) \in \ell \\ Q(m[(a+n):=v],\,\ell,\,t,\,k:: \leak (a+n)) }
{((x[n]=y),m,\ell,t,k) \Downarrow Q}

\inferrule[eval-input]
{\forall n.\;Q(m,\,\ell[x:=n],\,t :: \mathtt{IN }\;n, k)}
{((x=\texttt{input()}),m,\ell,t,k) \Downarrow Q}

\inferrule[eval-seq]
{(p_1,m,\ell,t,k) \Downarrow Q_1 \\
  (\forall m',\ell',t',k'.~ Q_1(m', \ell', t', k') \;{\Longrightarrow}\; (p_2,m',\ell',t',k') \Downarrow Q)}
{((p_1;p_2),m,\ell,t,k) \Downarrow Q}
\end{mathpar}

Observe that if a program that satisfies $\Downarrow$ reads an input and then stores to an address that depends on that input, all possible store addresses must be in the domain of $m$, and the postcondition $Q$ must accept all leakage traces (and memories) that may result. On the other hand, just reading an input and storing it to a fixed location $a[0]$ only requires $Q$ to admit one leakage trace: $[\leak a]$.

We should also emphasize that Bedrock2 embodies a rather different compiler-verification approach than, e.g., CompCert's~\cite{Leroy-backend}.
CompCert handles complex nondeterminism through two main mechanisms: (1) a determinizing high-level memory model and (2) small-step semantics with a bestiary of simulation-proof techniques.
Bedrock2 avoids the first complication by exposing a flat, machine code-style memory model at the source level.
More interestingly, Bedrock2's use of omnisemantics \emph{enables all compiler-phase proofs to be performed using forward reasoning and big-step hypotheses}, even with nondeterministic target languages~\cite[\S 6]{Omnisemantics}.
Such a proof proceeds by simple induction over the big-step derivation for the source program, deriving the matching derivation for the target program.
It was a nonobvious result that such a scheme works well for a language as expressive as Bedrock2, and with our work described in this paper, we add another surprising consequence: \emph{proof of information-flow hyperproperties not just in big-step, forward-reasoning proof style, but also with no explicit consideration of multiple executions of source or target programs}.

Our work inherits a limitation from Bedrock2: the semantics of Bedrock2 only describe terminating (totally correct) programs, because omnisemantics has not yet been adapted to intentional nontermination (e.g., infinite reactivity).
Thus, we consider nonterminating programs as out-of-scope.
However, we expect our techniques to extend easily to any form of omnisemantics augmented to support nontermination.
Further, even without omnisemantics support for nontermination, we should be able to lift our new results to infinitely reactive programs in the same way as did the original work with Bedrock2: by defining higher-level patterns like event loops~\cite{lightbulb}.

\subsection{RISC-V Semantics}

For the target language of compiler-correctness proofs, we instrumented the Bedrock2 {RISC-V} semantics~\cite{riscvsem} with leakage traces.
Like source-level Bedrock2 semantics, the RISC-V semantics is canonically interpreted as a weakest-precondition predicate transformer, which we modified also to log leakage.
Our semantics additionally considers all instruction-fetch addresses to be leaked.
\par Following is an excerpt of our function computing the leakage of an instruction.
\begin{verbatim}
| Add _ _ _ => Return LeakAdd (* no arguments to LeakAdd, so inputs remain secret *)
| Lw _ rs1 _ => addr <- getReg rs1; Return (LeakLw addr)
| Blt rs1 rs2 _ => a <- getReg rs1; b <- getReg rs2; Return (LeakBlt (word.lts a b))
\end{verbatim}

\section{Specifying Constant Time Using Partially Determinized Semantics}\label{sec:determinize}
In \cref{semapproach}, with the username-password example, we illustrated a sufficient condition for recognizing constant-timeness of programs: namely, a program is \emph{naive-constant-time} if its leakage is a function of public values (including both public initial state and public runtime input).
In \cref{sec:det}, we present examples of naive-constant-time specifications.
However, in \cref{stack_swap_bad} we show that some intuitively constant-time programs involving compiler-resolved nondeterminism fail to satisfy naive constant time.
Finally, in \cref{fix_stack_swap} we introduce the ``compile-time-determinized'' semantics \(\DownarrowA\) as an expedient solution, allowing us to formulate a more permissive definition.

\subsection{The Simple Case: No Leakage Dependence on Compiler-Resolved Nondeterminism}\label{sec:det}
In this section, we illustrate specifications of naive constant time.
The \(\Downarrow\) predicate we use in this section---and in the rest of this paper---is just like the one defined in \cref{bedrock2_sem}, except the postcondition takes only two arguments \((t, k)\) rather than four \((m,\ell,t,k)\).

\subsubsection{First example: \texttt{swap} of two values stored in memory}\label{sec:swap}
\begin{verbatim}
swap(int* pa, int* pb) { int tmp = *pa; *pa = *pb; *pb = tmp; }
\end{verbatim}
Consider the addresses \verb|pa| and \verb|pb| to be public, while the values \verb|*pa| and \verb|*pb| stored at the addresses are private.
Then we can specify naive-constant-timeness of \verb|swap| as follows.
\begin{align*}
  & \exists f.\; \forall a_\textrm{ptr}, b_\textrm{ptr}, a, b. \; \forall m, \ell. \\
  & (\texttt{swap(pa, pb)},\, m[a_\textrm{ptr} := a, b_\textrm{ptr} := b],\, \ell[\mathtt{pa} := a_\textrm{ptr}, \mathtt{pb} := b_\textrm{ptr}],\, [],\, []) \Downarrow \{([], f(a_\textrm{ptr}, b_\textrm{ptr}))\}.
\end{align*}
The final leakage is allowed to depend on the public values \(a_\textrm{ptr}, b_\textrm{ptr}\) but not the private values \(a, b,m,\ell\).
Thus, we interpret the specification above as saying that \(a,b,m,\ell\) are not leaked.
\par The leakage of \verb|swap| consists of a load at \(a_\textrm{ptr}\), a load at \(b_\textrm{ptr}\), a store at \(a_\textrm{ptr}\), and a store at \(b_\textrm{ptr}\).
So, to prove the specification above, we define \(f(a_\textrm{ptr}, b_\textrm{ptr}) := [\leak a_\textrm{ptr}, \leak b_\textrm{ptr}, \leak a_\textrm{ptr}, \leak b_\textrm{ptr}]\).

\subsubsection{Second example: \texttt{login}}\label{login} (written in Pythonish syntax)
\begin{verbatim}
def login():
  username = input("enter your username")
  password = lookup_password(username)
  attempt  = input("enter your password")
  print(if strs_eq(password, attempt) then "hooray" else "wrong")
\end{verbatim}
Consider the username to be public but the password to be private.
Thus a constant-time program should have leakage independent of both the password and the password attempt.
Here is a specification of the \verb|login| function, expressing that the leakage only depends on the username \(u\).
\[\exists f. \; \forall m,\ell. \; (\verb|login()|, m, \ell, [], []) \Downarrow \{([\verb|IN | u, \verb|IN | a, \verb|OUT | r], f(u)) : u,a,r \textrm{ are strings}\}.\]
\par However, the implementation of \verb|login| does not actually meet this specification.
The issue is that it branches on whether \verb|password == attempt|, and we assume branches are leaked!
So, it is unavoidable that the leakage depends on whether \verb|password == attempt|.
\par The following subtler specification actually holds for the \verb|login| function, assuming constant-time helper functions.
We assume some pure function \(L(m,u)\) that, given memory \(m\) and username \(u\), returns the output of the \verb|lookup_password| function on input \(u\) in \(m\).
\[\exists f. \; \forall m,\ell. \; (\verb|login()|, m, \ell, [], []) \Downarrow \{([\verb|IN | u, \verb|IN | a, \verb|OUT | r], f(u, a == L(m,u))) : u,a,r \textrm{ are strings}\}.\]
It was easy to specify that, although the leakage may not depend on the attempt \(a\) or the password \(L(m,u)\) in arbitrary ways, it is just allowed to depend on whether they are equal.

\subsection{Problem: Naive Constant Time Forbids Compiler-Resolved Nondeterminism}\label{stack_swap_bad}
\begin{verbatim}
  stack_swap() { int s[2] = {0, 1}; swap(&s[0], &s[1]); }
\end{verbatim}
Intuitively, \verb|stack_swap()| should be a well-defined constant-time program with a no-op behavior.
For \verb|stack_swap| to be naive-constant-time, its leakage should be a function of public values.
But its leakage depends on the address of the array \verb|s|---which we model as being nondeterministic---so in fact the leakage is not a deterministic function of anything, let alone public values.
We can still write a constant-time specification via an ad-hoc generalization of naive constant time.
In addition to being allowed to depend on public initial state and I/O, we permit the leakage to depend arbitrarily on the allocation address recorded earlier in the leakage trace via $\branchnoarg$:
\[\exists f. \; \forall m,\ell. \; (\verb|stack_swap()|, m, \ell, [], []) \Downarrow \{([], [\branch a] \app f(a)) : a \textrm{ is a word}\}.\]
For this program it was clear what to say: the trace should begin with a nondeterministic event, and then the rest of it is allowed to depend on the outcome of that event (but not private values!) arbitrarily.
However, more complicated programs could have $\branchnoarg$ events interleaved with $\mathtt{Leak}$ events, with the length of the trace depending on the outcomes of earlier $\branchnoarg$.
So, general constant-time specifications appear to require elaborate postconditions detailing possible allocation patterns throughout execution.
However, there is a simple alternative approach.


An alternative perspective on programs like \verb|stack_swap| is that the difficulty arises because the source-language semantics are too nondeterministic.
Our definition of \(\Downarrow\) says that the memory allocation is arbitrarily nondeterministic.
But then the compiler is given too much flexibility; rather, it should somehow be encoded in the semantics that memory allocation is independent of, for instance, secrets stored in memory.
To that end, the source-language semantics can instead say that memory allocation is some \emph{unspecified deterministic function} of values that, in particular, do not include secrets stored in memory.
We now expand upon this idea.

\subsection{Solution: Constraining Inputs to Compiler-Resolved Nondeterministic Choices}\label{fix_stack_swap}

If a source program's behavior depends on some nondeterministic choice resolved by the compiler, it is natural to ask what the compiler's choice is allowed in turn to depend on.
Conveniently, it appears that the answer for each unspecified behavior we considered in Bedrock2 is either ``anything'' or ``very little.''
For example, uninitialized \emph{contents} of freshly allocated memory could depend on anything (e.g., recently deallocated memory).
Thus, the programmer ought to avoid leaking these choices. 
In contrast, memory-allocation \emph{addresses} are necessarily leaked by source programs through load and store commands.
Thus, it is the compiler's responsibility to ensure that its method of address selection does not reveal the program's secrets.
Similar considerations arise when outputting (rather than leaking) nondeterminstically chosen values: revealing uninitialized memory contents could be compiled into an arbitrarily bad information leak, but printing a freshly allocated address is a common instrumentation technique---we will return to this aspect in \cref{sap_spec} and \cref{predictors_bad}.

The source-language semantics should specify what each compiler-implemented nondeterministic choice can depend on.
Concretely, there are no restrictions on the data-flow dependencies of uninitialized data, and allocation addresses can only depend on information that the program has already leaked---for example, past control flow.
Allowing memory-allocation addresses to depend arbitrarily on the leakage trace is a satisfying simplification: instead of separately keeping track of what could be leaked when the program reveals a choice made by compiler-generated code, the semantics consider everything that this choice could leak to be revealed eagerly.

More formally, our recommendation is to encode the compiler's limited permission to make an implementation choice as an opaque (top-level existentially quantified) \emph{oracle} function $\mathcal{A}$ in the semantics.
Each operational rule with limited nondeterminism calls \(\mathcal{A}\) with the leakage trace as an argument; the output of \(\mathcal{A}\) determines the outcome of the nondeterministic event.
This flow corresponds to the perspective that, when the dependencies of a choice are limited to a subset of the program's state, that choice is no longer nondeterministic. 

We will write \(\DownarrowA\) to denote a modified version of the Bedrock2 semantics, where each stack-allocation address is picked using an opaque function \(\mathcal{A}\) of the leakage trace.
The only difference between the two versions, \(\Downarrow\) and \(\DownarrowA\), is how the address \(a\) of a fresh stack allocation is picked:
\begin{mathpar}
\inferrule*[right=\texttt{stackalloc}]
{ \forall a.\;\;\;\;~\forall~m_{\textit{new}}.~(\mathsf{dom}~m_{\textit{new}} \cap \mathsf{dom}~m) = \varnothing{} \;\wedge\;
  \mathsf{dom}~m_{\textit{new}} = [a, a+n) \;\Longrightarrow\; \\\\
	(p,\,m \cup m_{\textit{new}},\, \ell[\texttt{x} := a],\,t,\,k :: \branch a) \Downarrow \{(m', \ell', t', k') : Q(m' - m_\textit{new},\ell',\,t',\,k')\}}
	{ ((\mathtt{stackalloc}~n~\texttt{as x;}~p),\,m,\,\ell,\,t,\, k) \Downarrow Q }

\inferrule*[right=\texttt{stackalloc}\(_\mathcal{A}\)]
	{ a = \mathcal{A}(k)\;\;\;\;
	\forall m_{\textit{new}}.~(\mathsf{dom}~m_{\textit{new}} \cap \mathsf{dom}~m) = \varnothing{} \;\wedge\;
  \mathsf{dom}~m_{\textit{new}} = [a, a+n) \;\Longrightarrow\; \\\\
	(p,\, m \cup m_{\textit{new}},\, \ell[\texttt{x} := a],\,t,\,k::\branch a) \DownarrowA \{(m', \ell', t', k') : Q(m' - m_\textit{new},\ell',\,t',\,k')\}}
	{ ((\mathtt{stackalloc}~n~\texttt{as x;}~p),\,m,\,\ell,\,t,\, k) \DownarrowA Q }
\end{mathpar}

We found this oracle-based encoding workable in both compiler and source-program proofs, but actually specifying the right semantics and understanding its implications were far from intuitive.
For example, our first attempt did not log oracle calls in the leakage trace, so we had accidentally written semantics guaranteeing that any two consecutive
(i.e., with no leakage events in-between)
allocations have the same address.
(Note that oracle calls are being logged with the \(\branch a\) event above, which serves the dual purpose of (1) logging the outcome of a nondeterministic event, which we will use in \cref{sec:two}, and (2) preventing the just-mentioned problem.)
\par Now that we have the \(\DownarrowA\) predicate, it is straightforward to define constant time for arbitrary Bedrock2 programs.
We say that a program is \emph{constant-time} if it executes with a leakage trace that is a function of (1) the oracle \(\mathcal{A}\) and (2) public values.

\subsubsection{Example: Now we can easily write a specification for \texttt{stack\_swap}}\label{stack_swap_spec}
\[\exists f. \; \forall \mathcal{A}. \; \forall m,\ell. \; (\verb|stack_swap()|, m, \ell, [], []) \DownarrowA \{([], f(\mathcal{A}))\}.\]
Leakage is allowed to depend on \(\mathcal{A}\) and public values (there are none) but not on secrets in $m, \ell$.

\subsubsection{Example: intentionally outputting compiler-resolved nondeterminism is safe per $\DownarrowA$}\label{sap_spec}
\begin{verbatim}
stackalloc_and_print() { stackalloc 1 as x; print(x); }
\end{verbatim}
Some security-related considerations are not captured by the leakage trace.
For instance, suppose an adversary can read values that are printed; then, the printed values should not depend on secrets.
\par As the program above prints the \emph{pointer} \verb|x| (as opposed to uninitialized memory such as \verb|*x|), we expect its output to be independent of secrets.
We formalize this expectation as follows; the point is that the printed value may depend on \(\mathcal{A}\), but it is independent of secrets in \(m\) and \(\ell\).
\[\exists f_\mathrm{io}, f. \; \forall \mathcal{A}. \; \forall m,\ell. \; (\verb|stackalloc_and_print()|, m, \ell, [], []) \DownarrowA \{([\verb|OUT | f_\mathrm{io}(\mathcal{A})], f(\mathcal{A}))\}.\]

\par The \verb|stackalloc_and_print| function is a toy example.
For a more realistic example where we care about outputs not depending on secrets, consider a concurrent program.
In Bedrock2, interaction with shared memory is modeled as I/O.
Instead of printing (as in \verb|stackalloc_and_print|), output might mean writing to shared memory. 
So, proving that outputs do not depend on secrets guarantees that the thread is not leaking secrets to other threads via memory.

\section{Specifying Constant Time Using Oracles in the Postcondition}\label{sec:two}
In \cref{sec:determinize}, we considered the problem of writing constant-time specifications in the context of compiler-resolved nondeterminism.
We presented a simple solution: get rid of the problematic nondeterminism.
This solution works well, but retaining nondeterminism can be convenient. 

This section describes how to achieve the same results as in \cref{sec:determinize} without resorting to the determinized semantics \(\DownarrowA\).
We want to find a more general solution than the ad-hoc approach used to write a specification of \verb|stack_swap| in \cref{stack_swap_bad}.
Finding a more general approach is important for two reasons.
First, writing custom specifications for each function would be complicated and error-prone.
Second, memory-allocation nondeterminism (unlike I/O) disappears after a program is compiled.
Thus, we need to treat it uniformly so that we can prove a compiler theorem about it.
(Roughly, the compiler theorem should say that if the source program is constant-time up to memory-allocation nondeterminism, then the target program is constant-time.)

\subsection{How to Write Constant-Time Specifications with Nondeterministic Semantics}\label{sec:ct_nondet}
Suppose a program \(p\) is constant-time.
Intuitively speaking, if we fix public values, then---as discussed in \cref{sec:determinize}---the leakage of \(p\) should be a function of the oracle \(\mathcal{A}\) with which \(p\) executes.
Our goal here, then, is to give a meaning to the phrase ``the oracle \(\mathcal{A}\) with which \(p\) executes,'' without resorting to the determinized \(\DownarrowA\).
This framing motivates the following definitions.

\par As a first step, we formalize how to say, in the context of \(\Downarrow\) rather than \(\DownarrowA\), that ``the nondeterminism was resolved according to \(\mathcal{A}\),'' or ``the leakage trace is compatible with \(\mathcal{A}\).''

\begin{definition}
  We say that a leakage trace \(k\) is \emph{compatible with} an oracle \(\mathcal{A}\), and we write \(k \sim \mathcal{A}\), if for all \(k_1, x, k_2\) such that \(k = k_1 \app [\branch x] \app k_2\), we have \(\mathcal{A}(k_1) = x\).
\end{definition}
The point is that, in a postcondition of \(\Downarrow\), the statement $k\sim\mathcal{A}$ holds if and only if \(\Downarrow\) happened to behave ``as if it were \(\DownarrowA\).''
By only requiring a postcondition to hold when \(\Downarrow\) behaves like \(\DownarrowA\), we can define a notion analogous to \(\DownarrowA\) directly in terms of $\Downarrow$.
\begin{definition}
  We say that \((p, m, \ell, [], []) \NdetDownarrowA Q\) if the postcondition \(Q\) holds whenever the program happens to execute compatibly with the oracle \(\mathcal{A}\).
  Formally,
  \[(p, m, \ell, [], []) \NdetDownarrowA Q := (p, m, \ell, [], []) \Downarrow \{(t, k) : k \sim \mathcal{A} \Rightarrow Q(t, k)\}.\]
\end{definition}

\subsection{Equivalence of Intrinsically and Postcondition-Wise Oracle-Based Semantics}\label{sec:equiv}
\begin{theorem}\label{equiv}
  For all \(p,m,\ell,Q\),
  \[\left[\forall \mathcal{A}. \; (p, m, \ell, [], []) \DownarrowA Q(\mathcal{A})\right] \iff \left[\forall \mathcal{A}. \; (p, m, \ell, [], []) \NdetDownarrowA Q(\mathcal{A})\right].\]
\end{theorem}
So, in a strong sense, the \(\NdetDownarrowA\) predicate is equivalent to the \(\DownarrowA\) predicate.
A surprising consequence is that \(\Downarrow\) can be used to express every specification---including constant-time specifications---that can be written with \(\DownarrowA\).
We do not need determinized semantics to define constant time.

\par \cref{sec:equiv2} proves this theorem.
The proof involves a few nonobvious technical devices but, we claim, does not depend on peculiarities of Bedrock2. 
More precisely: an analogous result should hold for any reasonable computer language with any nondeterministic construct being factored out into an oracle.
In particular, the proof does not rely on e.g. finiteness of the machine-state type in Bedrock2.
However---as with omnisemantics in general---it is not immediately obvious how to extend this result to fine-grained perpetually reactive behaviors.

\subsection{Example: A Constant-Time Specification of \texttt{stack\_swap} Using \(\NdetDownarrowA\)}
Mirroring~\cref{stack_swap_spec}, we can state that \verb|stack_swap| is constant-time:
\[\exists f. \; \forall \mathcal{A}. \; \forall m, \ell. \; (\verb|stack_swap()|, m, \ell, [], []) \NdetDownarrowA \{([], f(\mathcal{A}))\}.\]
It is a basic property of omnisemantics~\cite{Omnisemantics} that nonvacuous \(\forall\) quantifiers commute with $\Downarrow$.
In particular, by unfolding \(\NdetDownarrowA\), the above specification is equivalent to the following one:
\[\exists f. \; \forall m, \ell. \; (\verb|stack_swap()|, m, \ell, [], []) \Downarrow \{([], k) : \forall \mathcal{A}. \; k \sim \mathcal{A} \Rightarrow k = f(\mathcal{A})\}.\]

\section{Predictors: a Restricted Class of Functions Taking Oracles to Traces}\label{predictors}

Rephrasing the intuition we have been working from: a program \(p\) is constant-time if, after fixing public values, the trace of \(p\) is a function of the oracle \(\mathcal{A}\) with which \(p\) executes.
Thus there should be a procedure \(\mathcal{P}\) for taking the oracle \(\mathcal{A}\) with which \(p\) executes and obtaining the trace of \(p\).
\par Now, we consider what the procedure \(\mathcal{P}\) should look like.
It should not need to query the oracle arbitrarily.
For instance, \(\mathcal{P}\) should be able to generate a certain prefix of the trace---namely, the prefix before \(p\) takes any nondeterministic branches---without querying the oracle!
In general, given any prefix of the trace of \(p\), the procedure \(\mathcal{P}\) should be able to predict what the next event in the trace will be---except, when the next event is $\branchnoarg$, then \(\mathcal{P}\) cannot know which way the branch goes, so it has to query the oracle.
\par We will now give a formal definition corresponding to the class of procedures \(\mathcal{P}\) we are talking about.
As their function is to take a prefix of a trace and predict what comes next (modulo nondeterminism), we will call them \emph{predictors}.
We will say that \(\mathcal{P}\) \emph{predicts} a trace \(k\) if, for every oracle \(\mathcal{A}\) compatible with \(k\), the trace generated by running the procedure \(\mathcal{P}\) with the oracle \(\mathcal{A}\) will be \(k\).

\begin{definition}
  A \emph{predictor} is a function \(\mathcal{P}\) that takes in a trace and outputs one of the symbols \(\branchnoarg\), \(\leak x\), or \(\en\).
  We say \(\mathcal{P}\) \emph{predicts} a trace \(k\), and we write \(k \in \mathcal{P}\), if the following hold.
  \begin{itemize}
    \item $\forall k_1, k_2, x$, if  \(k = k_1 \app [\leak x] \app k_2\), $\;\;\;\;\;\;\;\;\;\;$ then \(\mathcal{P}(k_1) = \leak x\).
    \item $\forall k_1, k_2, x$, if \(k = k_1 \app [\branch x] \app k_2\), then \(\mathcal{P}(k_1) = \branchnoarg\).
    \item \(\mathcal{P}(k) = \en\).
  \end{itemize}
\end{definition}

\begin{definition}\label{ex_with_pred}
  Let \(\mathcal{P}\) be any function that takes I/O traces to predictors.
  We say that \((p, m, \ell, [], [])\) \emph{executes with predictor} \(\mathcal{P}\) if \((p, m, \ell, [], []) \Downarrow \{(t, k) : k \in \mathcal{P}(t)\}\).
\end{definition}
We say that a program \(p\) is \emph{predictor-constant-time} if it executes with a predictor depending only on public initial values and public runtime input.
More precisely: let \(g\) be such that \(g(t)\) comprises exactly the public information contained in the I/O trace \(t\) of \(p\).
Then \(p\) is \emph{predictor-constant-time} if it executes with a predictor of the form \(\mathcal{P} \circ g\), where \(\mathcal{P}\) depends only on public initial values.
\par Now, we show that this definition of predictors corresponds exactly to our intuitive thought of predictors as procedures that take oracles as input and return traces.
\begin{theorem}\label{pred_is_fun}
  For any predictor \(\mathcal{P}\), there exists a partial function \(f_\mathcal{P}\), which takes oracles to leakage traces, such that the following holds.
  \[\forall \mathcal{A}. \; \forall k. \; k \sim \mathcal{A} \implies \left[ k \in \mathcal{P} \Leftrightarrow f_\mathcal{P}(\mathcal{A}) = \mathtt{Some}\, k\right] \]
\end{theorem}
\begin{proof}
  The function \(f_\mathcal{P}(\mathcal{A})\) is precisely the procedure outlined in the first two paragraphs of this section.
  We describe it as an imperative program.
  The program \(f_\mathcal{P}(\mathcal{A})\) begins by setting the variable \(k \gets []\).
  Then, repeatedly: if the predicted next event \(\mathcal{P}(k)\) is \(\leak x\), then it sets \(k \gets k \app [\leak x]\);
  else if \(\mathcal{P}(k)\) is \(\branch\), then it sets \(k \gets k \app [\branch \mathcal{A}(k)]\);
  else if \(\mathcal{P}(k)\) is \(\en\), then it returns \(\mathtt{Some}\,k\).
  (If this program loops, we say it returns \(\mathtt{None}\).)
\end{proof}

\subsection{Equivalence of Predictor Constant Time and (Oracle) Constant Time}\label{pred_and_not}
Predictor constant time is a more structured and, a priori, stricter requirement than the notions of constant time developed in previous sections.
Specifically, if a program is predictor-constant-time, then not only is its trace a function \(f\) of the oracle with which it executes, but the function \(f\) takes a certain simple form: in particular, it comes from some predictor via \cref{pred_is_fun}.
\begin{corollary}\label{pred_to_oracle}
  Let \(g\) be any function.
  (As before, we interpret \(g(t)\) as being the public part of \(t\).)
  For any \(\mathcal{P}\), the following holds: for all \(p,m,\ell\),
  \[(p, m, \ell, [], []) \Downarrow \{(t, k) : k \in \mathcal{P} \circ g(t)\} \iff \forall \mathcal{A}.\; (p, m, \ell, [], []) \NdetDownarrowA \{(t, k) : \mathtt{Some }\,k = f_{\mathcal{P} \circ g(t)}(\mathcal{A})\}.\]
\end{corollary}
\begin{proof}
  Unfold \(\NdetDownarrowA\), and commute \(\mathcal{A}\) with \(\Downarrow\) to push \(\forall \mathcal{A}\) into the righthand postcondition.
  Since every trace is compatible with some oracle, the postconditions are equivalent by \cref{pred_is_fun}.
\end{proof}
If a program \(p\) is predictor-constant-time, then it satisfies a specification of the form of the left-hand side of \cref{pred_to_oracle}.
Therefore, it satisfies the specification on the right-hand side (which says that it is constant-time).
Thus, predictor constant time implies constant time.
The following theorem says the converse: if \(p\) is constant-time, then it is predictor-constant-time as well.
\begin{theorem}\label{oracle_to_pred}
  Let \(g\) and \(f\) be any functions.
  There exists some \(\mathcal{P}\) such that for all \(p, m, \ell\),
  \[\left[\forall \mathcal{A}.\; (p, m, \ell, [], []) \NdetDownarrowA \{(t, f(g(t), \mathcal{A}))\} \right]\implies (p, m, \ell, [], []) \Downarrow \{(t, k) : k \in \mathcal{P} \circ g(t)\}.\]
\end{theorem}
\begin{proof}
  This is significantly harder than \cref{pred_to_oracle}, in that the postcondition on the left-hand side does not actually imply the right-hand postcondition.
  The proof must actually use the structure of the \(\Downarrow\) predicate.
  However, the main idea is just to invert the procedure from \cref{pred_is_fun}.
\end{proof}

\subsection{Proving Programs Predictor-Constant-Time}\label{pppct}
\par For example, here is a specification saying that \verb|stack_swap| is predictor-constant-time.
\[\exists \mathcal{P}. \; \forall m,\ell. \; (\verb|stack_swap()|, m, \ell, [], []) \Downarrow \{([], k) : k \in \mathcal{P}\}.\]
Looking at the definition of \verb|stack_swap| (\cref{stack_swap_bad}) and the definition of \verb|swap| (\cref{sec:swap}), clearly the trace of \verb|stack_swap| will always be of the form $[\branch x;$ $\leak x;$ $\leak (x + 1);$ $\leak x;$ $\leak (x + 1)]$.
Then it is clear how to define \(\mathcal{P}\) to prove the specification of \verb|stack_swap|.

\par The full definition of \(\mathcal{P}\) would be awkward to write down here, but it is written in \cref{ss_pred_tree}.
Predictors tend to be unwieldy and unilluminating in any notation that follows the definition directly, which makes them inconvenient for source-program proofs.
On the other hand, we found them to be perfectly suitable for compiler proofs.
\cref{leakage_trees} discusses predictors' drawbacks in more detail and presents a nicer, mostly equivalent way to represent them.

Using oracles to model compiler-resolved nondeterminism naturally leads to compositional verification using standard omnisemantics program-logic constructions~\cite[\S 5]{Omnisemantics}.
The same is true for predictors, but it requires a trick that we detail in \cref{sec:nonmodular}.

\section{Constant-Time Compiler Specifications}\label{sec:compiler}
Constant-time specifications can be subtle, and it would not be straightforward to write a compiler theorem directly saying that every possible source-level constant-time theorem implies a corresponding target-level constant-time theorem.
So, our notion of compiler correctness will be intentionally more general than ``compiling constant-time programs to constant-time programs.''
\par To avoid distractions, we discuss compiler passes \(C : L \to L\) that work within a single language.

\subsection{The Simple Case: Compiler Correctness for Fully Deterministic Languages}
\begin{definition}\label{ctp1}
  \(C\) \emph{takes fixed-trace programs to fixed-trace programs} if
	\[\forall p.\;\exists\gamma_p.\;\forall k_H,m,\ell. \; (p, m, \ell, \text{\small{[]}}, \text{\small{[]}}) \Downarrow  \{(t, k) : k = k_H\} \implies (C(p), m, \ell, \text{\small{[]}}, \text{\small{[]}}) \Downarrow \{(t, k) : k = \gamma_p(k_H)\}.\]
\end{definition}
This way is the most direct to state that a compiler preserves deterministic constant time: high-level constant time is the hypothesis, and low-level constant time is the conclusion.
It is precisely analogous to the specification (Theorem 5.1) used for deterministic languages in \cite{CtCompCert}.
\par One can check that composing~\cref{ctp1} with our specification of the deterministic program \verb|swap| from \cref{sec:swap} yields a statement saying that \(C(\verb|swap()|)\) is constant-time.

\subsection{Compiler Correctness in the Presence of External Nondeterminism}\label{leakage_trans_defn}
Knowing that \(C\) takes fixed-trace programs to fixed-trace programs says nothing about its behavior on nondeterministic programs like \verb|login()|.
The issue is that \cref{ctp1} effectively has an implicit hypothesis: the source program can only have one possible leakage trace!
More precisely: since the leakage of \verb|login()| depends on runtime input, no single $k_H$ satisfying the hypothesis of \cref{ctp1} exists.

To support programs that perform I/O, we generalize the compiler specification to say \emph{unconditionally} that low-level leakage must be a deterministic function of high-level leakage.
\begin{definition}\label{ctp2}
  \(C\) \emph{admits leakage-transformation functions} if 
  \[\forall p.\, \exists \gamma_p.\, \forall Q, m, \ell. \; (p, m, \ell, [], []) \Downarrow Q \implies (C(p), m, \ell, [], []) \Downarrow \{(t, \gamma_p(k_H)) : Q(t, k_H)\}.\]
\end{definition}

If \(C\) admits a leakage-transformation function $\gamma_p$ for $p:=\verb|login|$, we can show that the specification of \verb|login| is preserved appropriately.
Applying \cref{ctp2} to the specification of \verb|login| written in \cref{login} yields the following low-level specification (where \(p := \verb|login()|\) for brevity).
\[\exists f. \; \forall m,\ell. \; (C(p), m, \ell, [], []) \Downarrow \{([\verb|IN | u, \verb|IN | a, \verb|OUT | r], \gamma_p \circ f(u, a == L(m,u))) : u,a,r \textrm{ are strings}\}.\]

\subsection{Compiler Correctness in the Presence of Compiler-Resolved Nondeterminism}\label{ltf_not_enough}\label{comparison}
As illustrated, requiring \(C\) to admit leakage-transformation functions guarantees that it preserves constant-time specifications, even when external nondeterminism is involved.
The reason is that the postcondition \(Q\) can encode a relation between \(k_H\) and the I/O log \(t\), and this relation remains meaningful on the low level because \emph{the compiler preserves \(t\)}, so we have common references to the same nondeterministic events.
In contrast, compiler-resolved nondeterministic events are not preserved by the compiler.
Thus, requiring \(C\) to admit leakage-transformation functions is no longer sufficient in the presence of compiler-resolved nondeterminism.
\par For example, in the case of \verb|stack_swap|, the high-level specification says that the leakage of \verb|stack_swap()| is a function of the high-level oracle, and the desired low-level specification is that the leakage of \(C(\verb|stack_swap()|)\) is a function of the low-level oracle.
If \(C\) merely admits leakage-transformation functions, then the leakage of \(C(\verb|stack_swap()|)\) depends only on the leakage of \verb|stack_swap()|---but the leakage of \verb|stack_swap()| depends in turn on the high-level oracle, so we cannot conclude that the leakage of \(C(\verb|stack_swap()|)\) depends only on the low-level oracle.
\par To conclude that the leakage of \(C(\verb|stack_swap()|)\) is a function of the low-level oracle, we must impose some additional constraint on the compiler.
There are two options.
(Note that the first is stricter; any compiler satisfying it automatically satisfies the second.)
\begin{itemize}[leftmargin=*]
\item[] \textbf{Specification 1:} The low-level leakage must only depend on the low-level oracle and the high-level leakage.
  In addition, the high-level oracle must depend only on the low-level oracle.
\item[] \textbf{Specification 2:} The function taking low-level oracle to low-level leakage must depend only on the function taking high-level oracle to high-level leakage.
\end{itemize}

If a compiler satisfies either Specification 1 or Specification 2, then it appropriately preserves the specification of \verb|stack_swap|.
However, Specifications 1 and 2 are subtly and significantly different.
Specification 1 presents deterministic resolution of stack-allocation addresses as an absolute assumption---arbitrary conclusions about the source program's behavior may depend on it.
\par For instance, suppose we want to show that our compiler appropriately preserves the specification of \verb|stackalloc_and_print|.
The conclusion follows easily for Specification 1: the specification of \verb|stackalloc_and_print| (\cref{sap_spec}) says that the printed value is a function of the high-level oracle, and then Specification 1 says that the high-level oracle is a function of the low-level oracle.
We conclude that the printed value is a function of the low-level oracle.
\par In contrast, Specification 2 does not imply that the compiler preserves the specification of \verb|stackalloc_and_print|.
Specification 2 only allows us to make conclusions about leakage, and the specification of \verb|stackalloc_and_print| has nothing to do with leakage.

\subsection{Formalizing Specification 1: Requiring Oracle-Transformation Functions}
\begin{definition}\label{ctp5}
  \(C\) \emph{admits leakage-transformation and oracle-transformation functions} if
  \begin{align*}
    \forall p. \; \exists \gamma_p,\mathcal{A}_p. \; \; \; \; \; \; \; \forall Q, m, \ell. \; & (\forall \mathcal{A}. \; (p, m, \ell, [], []) \DownarrowA Q(\mathcal{A})) \implies \\
    & (\forall \mathcal{A}. \; (C(p), m, \ell, [], []) \DownarrowA \{(t, \gamma_p(k_H, \mathcal{A})) : Q(\mathcal{A}_p(\mathcal{A}))(t, k_H)\}).
  \end{align*}
\end{definition}
The key points are that high-level oracle \(\mathcal{A}_p(\mathcal{A})\) depends only on low-level oracle \(\mathcal{A}\), and low-level leakage \(\gamma_p(k_H, \mathcal{A})\) depends only on high-level leakage \(k_H\) and low-level oracle \(\mathcal{A}\).
Note: by \cref{equiv}, we could replace \(\DownarrowA\) with \(\NdetDownarrowA\) in \cref{ctp5}, and the definition would be equivalent.
\par In the case of passes with different source and target languages, \cref{ctp5} is still the correct notion; the low-level oracle just might look different.
For instance, in the degenerate case where the low-level language is deterministic, we take the low-level oracle to have type \verb|unit|.
This pattern applies to the Bedrock2 compiler phase that first produces machine code---on the level of machine code, there is no stack-allocation nondeterminism and hence no oracle.

\subsubsection{Example: \texttt{stack\_swap}}\label{sec:stack_swap_ex}
Applying \cref{ctp5} to the specification of \verb|stack_swap| from \cref{stack_swap_spec} (and letting \(p := \verb|stack_swap()|\) for brevity), we get an appropriate low-level specification:
\[\exists f. \; \forall \mathcal{A}. \; \forall m,\ell. \; (C(p), m, \ell, [], []) \DownarrowA \{([], \gamma_p(f \circ \mathcal{A}_p(\mathcal{A}), \mathcal{A}))\}.\]

\subsubsection{Example: \texttt{stackalloc\_and\_print()}}\label{sec:compile_sap}
Applying \cref{ctp5} to the specification from \cref{sap_spec} of \verb|stackalloc_and_print| (and letting \(p := \verb|stackalloc_and_print()|\)), we get
\[\exists f_\mathrm{io}, f. \; \forall \mathcal{A}. \; \forall m,\ell. \; (C(p), m, \ell, [], []) \DownarrowA \{([\verb|OUT | f_\mathrm{io} \circ \mathcal{A}_p(\mathcal{A})], \gamma_p(f \circ \mathcal{A}_p(\mathcal{A}), \mathcal{A}))\}.\]

\subsection{Formalizing Specification 2: Requiring Predictor-Transformation Functions}
Without assuming that the high-level oracle is a function of the low-level oracle, it seems very tricky to prove Specification 2.
The problem is that, a priori, the function taking high-level oracle to high-level leakage could be arbitrarily complicated.
For instance, suppose the high-level leakage is \([]\) when \(\mathcal{A}([\leak 0] * 1000) = 0\) and \([\leak 0]\) otherwise.
It is unclear how to construct a corresponding function taking low-level oracle to low-level leakage.
To make Specification 2 feasible to prove, we must restrict how the high-level leakage can depend on the high-level oracle.
\par There is a natural solution: the high-level leakage should only be allowed to depend on the high-level oracle via a predictor!
That is, our specification should say that there exists a \(\gamma_p\) such that for any \(\mathcal{P}\), if the high-level leakage is the function \(f_\mathcal{P}\) of the high-level oracle, then the low-level leakage is the function \(f_{\gamma_p(\mathcal{P})}\) of the low-level oracle.
\begin{definition}\label{ctp3}
  \(C\) \emph{admits predictor-transformation functions} if 
  \begin{align*}
    \forall p. \; \exists \gamma_p. \; \; \; \; \; \; \; \forall Q, m, \ell. \; & (p, m, \ell, [], []) \Downarrow Q \implies \\
    & (C(p), m, \ell, [], []) \Downarrow \{(t, k_L) : \exists k_H. \; Q(t, k_H) \land (\forall \mathcal{P}. \; k_H \in \mathcal{P} \Rightarrow k_L \in \gamma_p(\mathcal{P}))\}.
  \end{align*}
\end{definition}
Just as \cref{ctp2} says that the low-level trace is a function of the high-level trace, \cref{ctp3} says that the low-level predictor is a function of the high-level predictor.

We emphasize that \cref{ctp3} differs from \cref{ctp5} along \emph{two} dimensions.
First, \cref{ctp3} is an instance of Specification 2 rather than Specification 1.
This difference is unrelated to predictors; one could formalize Specification 2 without mentioning predictors.
(As discussed, such a formalization seems infeasible to prove.)
Second, \cref{ctp3} is (a priori) a restricted form of Specification 2: it only applies in the case where the high-level leakage depends on the high-level oracle via a predictor, not via an arbitrary function.
However, by \cref{pred_and_not}, the high-level leakage depends on the high-level oracle via an arbitrary function \emph{if and only if} it depends on the high-level oracle via a predictor.
That is: since predictor constant time is equivalent to constant time, and \cref{ctp3} says that the compiler preserves predictor constant time, \cref{ctp3} also says that the compiler preserves constant time.

\par Just like \cref{ctp5}, \cref{ctp3} generalizes straightforwardly to interlanguage passes.
In the special case where there is no compiler-resolved nondeterminism on the low level, specifying a low-level predictor is equivalent to specifying a low-level leakage trace; so in this special case, \cref{ctp3} states that the low-level \emph{leakage} is a function of the high-level predictor. 
\subsubsection{Example: \texttt{stack\_swap}}
Assume that \(C\) admits predictor-transformation functions.
Then, applying \cref{ctp3} to the specification of \verb|stack_swap| from \cref{pppct} (and letting \(p := \verb|stack_swap()|\) for brevity), we obtain the following low-level specification.
\[\exists \mathcal{P}. \; \forall m, \ell. \; (C(p), m, \ell, [], []) \Downarrow \{([], k_L) : \exists k_H. \; k_H \in \mathcal{P} \land (\forall \mathcal{P}_0. \; k_H \in \mathcal{P}_0 \Rightarrow k_L \in \gamma_{p}(\mathcal{P}_0))\}.\]
The above postcondition clearly implies the simpler postcondition \(\{([], k_L) : k_L \in \gamma_p(\mathcal{P})\}\), which is exactly the statement that \(C(p)\) is predictor-constant-time.

\subsection{Extra Flexibility Afforded to the Programmer by Specification 1}\label{predictors_bad}
As discussed in \cref{comparison}, Specification 1 allows the source-level programmer to assume that a program will actually execute with some oracle, whereas Specification 2 makes no such promise!
As we saw with \verb|stackalloc_and_print|, Specification 2 only provides guarantees about leakage traces---on the other hand, Specification 1 allows us to guarantee that observable behavior of a program is independent of secrets, \emph{even when that behavior is not captured in the leakage trace}.

\subsection{Extra Flexibility Afforded to the Compiler by Specification 2}\label{predictors_good}
Specification 2 gives the compiler more freedom: compiler-resolved nondeterministic events \emph{need not actually be a function of previous leakage trace}.
Therefore, certain reordering optimizations that a compiler could perform violate Specification 1 but nevertheless satisfy Specification 2.
\par The problem with reordering optimizations is the following.
If the compiler reorders some parts of a program, then some nondeterministic event \(N\) may end up being after some event \(E\) in the target program, whereas in the source program \(N\) came before \(E\) and hence the outcome of \(N\) was not allowed to depend on \(E\).
A compiler satisfying Specification 1, then, is not permitted to perform such reordering in the case that on the low level, the outcome of \(N\) would depend on \(E\).
\par As a concrete example, Specification 1 does not permit a compiler to transform \(p\) into \(p'\).
\begin{alltt}
\(p\,\,:=\) \{ random as x; z = *w; print(x); \}
\(p':=\) \{ z = *w; random as x; print(x); \}
\end{alltt}
\par Here, for simplicity, we replace the Bedrock2 construct \verb|stackalloc as x| with the \verb|random as x| construct, which works exactly the same way in that it binds \verb|x| to a random number; the only difference is that it does no memory allocation. 
In \cref{app:compile-predictor}, we prove that a compiler transforming \(p\) into \(p'\) can satisfy Specification 2 (\cref{ctp3}) but not Specification 1 (\cref{ctp5}).

\subsection{Variations on Specification 2}
Specification 2 allows the compiler the extra freedom detailed in \cref{predictors_good}.
However, as discussed in \cref{predictors_bad}, Specification 2 only talks about functions taking oracles to \emph{leakage traces}, and therefore Specification 2 is only helpful when you want to prove things about leakage traces.
\par If we wanted to make Specification 2 more useful, we could have it talk about functions taking oracles to other things besides leakage traces.
An easy way to do this is just to modify our definition of leakage trace so that the ``other things'' that we are interested in appear in the leakage trace.
\par For example, to make Specification 2 preserve the specification of \verb|stackalloc_and_print|, we could parameterize the semantics \(\Downarrow\) over a function that says, given an I/O event, what should be the corresponding events (if any) added to the leakage trace.
By putting all the relevant information (and no more) in the leakage trace, we ensure that every theorem we care to prove about the low-level program is of the form ``the low-level program executes with some particular predictor,'' and hence Specification 2 is, from the programmer's perspective, equally useful to Specification 1.
\par This approach sounds inconvenient but workable.
Bedrock2 semantics are already parameterized over an I/O specification~\cite{lightbulb}.
Augmenting this specification to say which inputs and outputs are private---and then writing a compiler specification in the style of Specification 2---would allow for the benefits of both Specification 2 and Specification 1.
Thus, variations on \cref{ctp3} could allow for more useful compiler specifications than either of the specifications (\cref{ctp5} and \cref{ctp3}) we have presented here.
We leave the details for future work.

\section{Compiler Proof Techniques}\label{compiler_proof}
In this section, we outline how we proved that the Bedrock2 compiler satisfies the specifications we have discussed.
The compiler was already constant time-preserving; we did not need to modify it.
\par We completed proofs that the Bedrock2 compiler satisfies each of \cref{ctp5} and \cref{ctp3}.
In each case, we verified the usefulness of the specification by applying it to source-level constant-time theorems to yield assembly-level constant-time theorems.
We proved \cref{ctp5} in both the form that it is stated as well as the equivalent form in terms of \(\NdetDownarrowA\) rather than \(\DownarrowA\).
The two proofs were similar; one proceeded by induction on \(\Downarrow\) and the other by induction on \(\DownarrowA\).
\par In~\cref{wlt}, we show how to construct leakage-, predictor-, and oracle-transformation functions.
We wrote these directly in Gallina (Coq's dependently typed language) using well-founded recursion.
Then, we proved the correctness of the transformation functions (e.g, that \(\gamma_p\) and \(\mathcal{A}_p\) satisfy \cref{ctp5}) by induction on \((p, m, \ell, t, k) \Downarrow Q\), applying the introduction rules of $\Downarrow$ to symbolically execute each statement generated by the compiler pass $C$.
Detailed examples and explanations of such proofs for Bedrock2 and lambda calculi can be found in~\citet[\S 6]{Omnisemantics}.
\par We emphasize that omnisemantics makes compiler proofs easy, and our leakage-related additions did not require any changes to the structure of the Bedrock2 compiler proofs.

\subsection{Simple Case: Leakage-Preserving Optimizations}\label{leakage_preserving}
Some optimization passes are \emph{leakage-preserving}: the target program always has the same leakage as the source program.
This strong property straightforwardly implies whichever compiler specification we want to prove; the leakage-/predictor-/oracle-tranformation functions are identities.

\subsection{Constructing Transformation Functions}\label{wlt}

\paragraph{Leakage-Transformation Functions $\gamma_p$}
The idea is that knowing the high-level leakage (especially, knowing which way branches go) allows us to simulate the execution of the source program, even though we know nothing about the initial state of the source program.
Then, since we know the function \(C\) taking the source program to the target program, knowing source-level control flow yields knowledge of target-level control flow and in turn target-level leakage.
\par Concretely, a leakage-transformation function walks through the execution of the source program, guided by the source leakage trace telling it which branches to take.
For each statement executed in the source program, it looks at the leakage of that statement and computes the corresponding target-level leakage.
For example, the address of a source-language load appears in the source-language leakage trace and is reinserted into the target-language leakage trace if the compiled code also contains a corresponding load.
On the other hand, addresses of target-level loads to stack slots only depend on the stack-frame address, which is known from control flow.

\paragraph{Predictor-Transformation Functions $\gamma_p$}
To construct a predictor-transformation function, we assume given a high-level predictor \(\mathcal{P}_H\), with which the source program executes; and a low-level leakage trace \(k_L\), which is a prefix of the target-program leakage trace.
The requirement is to predict what comes after \(k_L\).
To do so, the predictor-transformation function walks the source program very similarly to a leakage-transformation function, keeping track of the high-level leakage trace as it goes and using \(\mathcal{P}_H\) to see what comes next.
When it has simulated the target program far enough to get to the end of \(k_L\), it checks what event comes next based on control flow of the target program.
If the next event is compiler-resolved nondeterminism, \(\branchnoarg\) is returned, otherwise the event itself (possibly by querying \(\mathcal{P}_H\) to get the corresponding high-level event).

\paragraph{Oracle-Transformation Functions $\mathcal{A}_p$}\label{wlot}
To construct an oracle-transformation function, we assume given a low-level oracle \(\mathcal{A}_L\), with which the target program executes; and a high-level leakage trace \(k_H\), which is a prefix of the source-program leakage trace.
The requirement is to return the outcome of the high-level compiler-resolved nondeterminism that comes after \(k_H\).
To do so, an oracle-transformation function walks through the source and target programs, guided by \(k_H\) and \(\mathcal{A}_L\).
When it reaches the end of \(k_H\), it looks at the next instance of compiler-resolved nondeterminism in the source program and determines the outcome of this choice by consulting the current low-level state (e.g., the value of a stack pointer or frame pointer) and the low-level oracle \(\mathcal{A}_L\).

\subsection{Transformation-Function Examples}\label{trans_fun}
\subsubsection{Introducing Nondeterminism in the Spilling Phase}\label{more_ndet}
Here we illustrate how to construct the function \(\gamma_p\) of \cref{ctp5} when \(C\) is the spilling phase of our compiler.
The spilling phase consists of a pair of functions \verb|spill_stmt| and \verb|spill_fun|.
The former compiles a statement.
The latter takes as input a function---that is, a list of argument names, a list of return-value names, and a function body---and returns the body of the compiled function.
\begin{verbatim}
Definition spill_fun argnames resnames body (words_needed : Z) : stmt :=
  stackalloc (bytes_per_word * words_needed) as fp;
  set_vars_to_reg_range argnames a0;
  spill_stmt body;
  set_reg_range_to_vars a0 resnames.
\end{verbatim}

We have corresponding leakage-transformation functions \verb|leak_stmt| and \verb|leak_fun|.
The leakage of a function body will depend on the value of the frame pointer \verb|fp|.
So, the \verb|leak_stmt| function needs to take the following inputs: the low-level oracle \verb|AL|, the high-level program \verb|sH| and its leakage \verb|kH|, the current value \verb|fpval| of the dedicated frame-pointer register \verb|fp|, and the low-level leakage that has accumulated so far \verb|kL_so_far|.
It will return the total low-level leakage that will have accumulated after \verb|sH| finishes executing.
In Coq:
\begin{verbatim}
Definition leak_stmt (AL : leakage -> word) (sH : stmt) (kH : leakage)
(fpval : word) (kL_so_far : leakage) : leakage := ...
\end{verbatim}
We consider how to implement \verb|leak_fun| given \verb|leak_stmt|.
The specification of \verb|leak_fun| is as follows: given a high-level function (\verb|argnames|, \verb|resnames|, \verb|body|), a low-level oracle \verb|AL|, the high-level leakage \verb|kH|, and the low-level leakage-so-far \verb|kL_so_far|, return the total leakage accumulated after the body of the spilled function finishes executing.
It is implemented as follows.

\begin{verbatim}
Definition leak_fun argnames resnames body (AL : leakage -> word) (kH : leakage)
    (kL_so_far : leakage) : leakage :=
  let fpval := AL kL_so_far in
  let kL' := [CompNonDet fpval] ++ leak_set_vars_to_reg_range fpval argnames in
  leak_stmt AL body kH fpval (kL_so_far ++ kL')
  ++ leak_set_reg_range_to_vars fpval resnames.
\end{verbatim}
To understand what is going on in \verb|leak_fun|, compare it to \verb|spill_fun|.
The function \verb|leak_fun| begins by querying the low-level oracle to see what value is put in the register \verb|fp|.
Then, it computes the leakage of the stackalloc call (which is \(\branch\verb|fpval|\)) and the leakage of the \verb|set_vars_to_reg_range| operation.
It passes the results to \verb|leak_stmt|, which returns the leakage that has accumulated after \verb|spill_stmt body| executes.
Then it finishes by appending the leakage of the \verb|set_reg_range_to_vars| operation.

\subsubsection{Resolving Nondeterminism in the FlatToRiscv Phase}
Here we illustrate how to construct the functions \(\mathcal{A}_p\) and \(\gamma_p\) of \cref{ctp5} when \(C\) is the FlatToRiscv phase of our compiler.
This phase consists of a function \verb|compile_stmt|, which takes as input \verb|mypos|, the relative position of the output code relative to a base position; \verb|stackoffset|, a value such that \verb|stackoffset + sp_val| is the highest used stack address (where \verb|sp_val| is the current value in the stack-pointer register \verb|sp|); and a statement \verb|sH| to be compiled.
It outputs a list of assembly instructions.
\begin{verbatim}
Fixpoint compile_stmt(mypos: Z)(stackoffset: Z)(sH: stmt): list Instruction :=
match sH with
| stackalloc n as x; body =>
  [Addi x sp (stackoffset-n)] ++ compile_stmt (mypos + 4) (stackoffset-n) body
... (* other cases omitted *) end.
\end{verbatim}
We have a corresponding leakage-transformation function \verb|leak_stmt|, which takes the following inputs: the high-level program \verb|sH| with leakage trace \verb|kH|, the values \verb|mypos| and \verb|stackoffset| (which have the same meanings as before), and the current value \verb|sp_val| of the stack-pointer register \verb|sp|.
It outputs the leakage of the compiled program.
Note that there is no compiler-resolved nondeterminism on the low level (i.e., the low-level oracle has type \verb|unit|), so \verb|leak_stmt| (unlike the \verb|leak_stmt| of \cref{more_ndet}!) does not need to take a low-level oracle as input.

\begin{verbatim}
Fixpoint leak_stmt (sH : stmt) (kH : leakage) (mypos stackoffset sp_val : word)
: list LeakageEvent :=
match sH with
| stackalloc n as x; body =>
  match kH with
  | CompNonDet _ :: kH' =>
    [ LeakAddi ] ++ leak_stmt body kH' (mypos + 4) (stackoffset - n) sp_val)
  | _ => (*should be impossible! kH is the leakage of an execution of sH*)
  end ... (*other cases omitted*) end.
\end{verbatim}
To understand \verb|leak_stmt|, compare it to \verb|compile_stmt|.
The function \verb|leak_stmt| begins by breaking the high-level leakage \verb|kH| into two pieces: the leakage of the line \verb|stackalloc n as x|, which is \verb|CompNonDet _|; and the leakage of the \verb|body| that comes after, which is \verb|kH'|.
\par We also have an \verb|oracle_stmt| function, which takes the same inputs as \verb|leak_stmt|; the only difference is that with \verb|leak_stmt| the input leakage \verb|kH| was expected to be the leakage of the whole high-level program \verb|sH|, whereas with \verb|oracle_stmt|, it is only expected to be some prefix of the leakage of \verb|sH|.
Then \verb|oracle_stmt| returns the output of the high-level oracle when given \verb|kH|.

\begin{verbatim}
Fixpoint oracle_stmt sH kH mypos stackoffset sp_val : word :=
match sH with
| stackalloc n as x; body =>
  match kH with
  | CompNonDet _ :: kH' => oracle_stmt body kH' (mypos+4) (stackoffset-n) sp_val
  | [] => sp_val + (stackoffset - n)
  | _ => (*this case should be impossible!*)
  end ... (*other cases omitted*) end.
\end{verbatim}
If \verb|kH| is the empty leakage trace, then \verb|oracle_stmt| has finished simulating the source and target programs, and it should simply output the value bound to \verb|x| by the high-level \verb|stackalloc| call.
Looking at \verb|compile_stmt|, we see that the value bound to \verb|x| is \verb|sp_val + (stackoffset - n)|, so it is what \verb|oracle_stmt| returns.
On the other hand, if \verb|kH| is not the empty leakage trace, then \verb|oracle_stmt| is not yet done simulating the source and target programs, so it passes the appropriate values to the recursive call and continues by simulating \verb|body|.

\subsection{A Summary of Our Compiler Proof Effort}\label{compiler_effort}

The whole project, to the point of submitting this paper, took about 12 months.
Most of the time was spent trying to find the right specification of the compiler to handle stack allocation and trying to find the right strategy for writing the compiler proof.
The adaptation of the Bedrock2 compiler proof was done by one person, who spent 18 hours per week on the project over 38 weeks.
\par In total, we experimented with three approaches to verifying that the Bedrock2 compiler is constant time-preserving.
We proved, by induction on \(\DownarrowA\), that it satisfies the version of \cref{ctp5} written in terms of \(\DownarrowA\).
Separately, we proved, by induction on \(\Downarrow\), that it satisfies the version of \cref{ctp5} written in terms of \(\NdetDownarrowA\).
Finally, we proved, by induction on \(\Downarrow\), that it satisfies \cref{ctp3}.
All three approaches required a similar amount of effort.
\par \cref{the_table} summarizes code written per compiler pass to prove \cref{ctp5} by induction on \(\DownarrowA\).
Since leakage- and oracle-transformation functions are structured similarly, we did not write them separately, instead writing one function returning a tuple (leakage, oracle output).
\begin{table}
  \caption{Lines of code written to prove the Bedrock2 compiler satisfies \cref{ctp5}}
  \label{the_table}
\begin{tabular}{| c | c | c | c | c | c |}\hline
  Pass & Base (proof lines before) & Proof lines added, deleted & Transf.-func. lines \\\hline
  FlattenExpr   & 956             & + 48, - 44  &       0          \\\hline
  RegAlloc      & 1353             &  + 10, - 8 &     0           \\\hline
  UseImmediate  & 181              & + 18, - 26 & 0               \\\hline
  DeadCodeElim  &     386              &  + 140, - 160  & 222               \\\hline
  Spilling      &   2034         & + 292, - 236 &        300      \\\hline
  FlatToRiscv   & 4793             & + 426, - 131 &      484      \\\hline
  composition   &  1499           & + 416, - 164  &      0          \\\hline
\end{tabular}
\end{table}

\par The proof work consisted mainly of small edits to existing proof scripts.
Note that most of the code is in the bodies of the leakage-/oracle-transformation functions.
Adapting the FlattenExpr, RegAlloc, and UseImmediate phases was little work, as they are leakage-preserving (\cref{leakage_preserving}).
\par We also encountered a lucky natural experiment to let us measure directly effort to adapt a new phase: the DeadCodeElim phase was added to the Bedrock2 compiler during our project, and we adapted that phase very recently.
Having the other phase proofs to reference, adapting the DeadCodeElim phase took only one day (about 5 hours).


\section{Source-Program Case Studies}\label{case_studies}

The case studies described in this section followed the semantics-determinization approach with \(\DownarrowA\) (\cref{fix_stack_swap}).
All case studies are included in our Coq implementation.

\subsection{Common Compiler Gotchas for Constant-Time Cryptography}

\paragraph{Division by Constants}
Compilers choosing non-constant-time implementation strategies for straightforward source-level constructs has led to numerous security issues in cryptographic software.
Two recent and straightforward examples are the \emph{lower} compiler optimization generating non-constant-time machine instructions for division and modulo by a constant in the NIST-PQC-competition algorithms Kyber and HQC, leading to exploitable vulnerabilities~\cite{KyberSlash,DivideAndSurrender}.
Our semantics also afford the compiler this flexibility: division and modulo leak their arguments.
However, the recommended fix (choosing the right implementation strategy at the source level) can be proven constant-time using our semantics.
Specifically, we transcribed the fixed \verb|poly_tomsg| to Bedrock2 and proved that its leakage trace depends only on the memory addresses of the input and output, not the scalar values in these arrays.

\paragraph{Constant-Time Memory Comparison}

In most languages, unexpectedly clever compiler optimizations can detect functional behavior of constant-time idioms and replace them with (average-case faster) versions that leak information.
A recurring example of this phenomenon, and a staple constant-time function found in every good cryptography library's ``subtle'' section, is comparing two buffers for equality without leaking anything else about their contents.
(Naively returning ``not equal'' on the first mismatch would leak the length of the common prefix, allowing for incremental guessing.)
We proved that the \verb|memequal| function added to the Bedrock2 standard library as a part of the Bedrock2-Fiat-Crypto integration~\cite{garagedoor} indeed does not leak its in-memory inputs:

\begin{verbatim}
Definition memequal := func! (x,y,n) ~> r { r = $0;
  while n {  r = r | (load1(x) ^ load1(y)); x = x + $1; y = y + $1; n = n - $1 };
  r = (r == $0) }.
\end{verbatim}

We proved a source-level specification of \verb|memequal|.
Then we compiled it and applied a compiler theorem (following \cref{ctp5}), obtaining the following (abridged) assembly-level specification.
\begin{verbatim}
Lemma memequal_ct : forall x y n pos stack_pointer ret_addr,
  exists finalK, forall (xs ys : list word) stack_space initialMachine,
  (*hypothesis saying xs and ys are length-n arrays at addresses x and y*) ->
  (*hypothesis saying stack_space is big enough for memequal to run*) ->
  initialMachine.(getPc) = pos /\ initialMachine.(getLeakage) = [] ->
  map.get initialMachine.(getRegs) ra = Some ret_addr ->
  arg_regs_contain initialMachine.(getRegs) [x; y; n] ->
  (*hypothesis saying the instructions of memequal are at address pos*) ->
  runsTo initialMachine (fun finalMachine : RiscvMachine =>
    finalMachine.(getPc) = ret_addr /\ finalMachine.(getLeakage) = finalK).
\end{verbatim}
Note that the leakage \verb|finalK| is independent of, for instance, the content \verb|xs| and \verb|ys| of the arrays at locations \verb|x| and \verb|y|.
In fact, \verb|finalK| depends only on the arguments \verb|x|, \verb|y|, \verb|n|, the address \verb|pos| of the compiled \verb|memequal| function, the stack pointer \verb|stack_pointer|, and the return address \verb|ret_addr| (leaked when the machine jumps back there).

Asking Coq to \verb|Print Assumptions memequal_ct| yields only propositional and functional extensionality.
The only other trusted code here (other than Coq) is our leakage-augmented RISC-V semantics, encapsulated in the \verb|runsTo| predicate visible in the code above.

\subsection{Password-Based Login}

Our next software case study is a password prompt inspired by \verb|agetty| used for Linux console login.
The program reads a password character-by-character and compares it against a reference value, returning whether the password was correct, without leaking either operand. 
While the conceptual task of this example is just to call \verb|memequal|, including some application context helps illustrate the different types of nondeterminism supported by our updated Bedrock2 compiler.

First, the program stack-allocates a temporary buffer for the user input.
Per $\DownarrowA$, the allocated address is independent of the password stored in program memory.
This fine point is important, as \verb|getline| writing the entered password to that buffer leaks the address of the buffer.

Second, \verb|getline| reads the user input by calling \verb|getchar| until a newline is entered.
As a result, the program branches on every character of the entered password compared with a newline!
Nevertheless, we prove that the leakage of \verb|getline| only depends on the \emph{length} of the input \verb|bs|, as well as the length \verb|n| and address \verb|dst| of the destination memory buffer.
More specifically: for any \(\mathcal{A}\),
\[
	\exists f.\; \forall m, \ell. \;
	(\mathtt{getline(dst, n)}, m, \ell, [], [])
	\DownarrowA
	\{(\mathtt{getline\_io}(\mathtt{bs}), f(|\mathtt{bs}|, \mathtt{dst}, \mathtt{n})) \mid \textrm{passwords } \mathtt{bs}\}.
\]
Note that the distinction between values of \verb|bs| (which are secret) and their lengths is program-specific, not understood by the compiler.
Nevertheless, preserving leakage-free execution during compilation is a genuine obligation: replacing each character-newline comparison with a 256-entry jump table as a misguided optimization would leak every character's value, not just the length.

Finally, the login program calls \verb|memequal| to compare the entered and actual passwords.
The overall postcondition establishes that the return value is \verb|1| iff the password from \verb|getline| is correct, but the leakage trace can be predicted based on memory addresses and lengths alone.

To further exercise corner cases of the semantics, we also proved a version where \verb|memequal| is used to compare the correct password against the entire buffer passed to \verb|getline|, not just the initialized part. In Bedrock2, unlike ISO C, branching on uninitialized values is defined behavior, so this program is legal, and the postcondition of \verb|memequal| still guarantees that the program does not leak any information about the password, even though the uninitialized bytes may depend on it.

\begin{verbatim}
Definition password_checker := func! (password) ~> ret {
  stackalloc 8 as x; (* password is 8 characters *) unpack! n = getline(x, $8);
  unpack! ok = memequal(x, password, $8); ret = (n == $8) & ok }.
\end{verbatim}

Implementing and proving this example (including \verb|getline|) took less than one workday, with less than an hour spent on leakage proofs.

\subsection{Output Without Leaking: Semiprime Generator}

We opted not to consider program inputs or outputs as leaked, since such direct information flows can be ruled out in the original Bedrock2 program logic.
This choice also gives us a chance to model cryptographic code in a useful way, where we assume an adversary does not have the computational power to e.g. factor a large secret key that has been output.
As an example of this strategy, we proved that the following program that reads two prime numbers as runtime input and outputs their product does not leak the factors:
\begin{verbatim}
Definition semiprime := func! () ~> (p, q) {
  p = getprime(); q = getprime(); n = p * q; output(n) }.
\end{verbatim}
Specifically, we established the following specification: for any \(\mathcal{A}\),
\[
  \exists f. \; \forall m, \ell. \; (\verb|semiprime()|, m, \ell, [], []) \DownarrowA \{([\mathtt{getprime}(p), \;\mathtt{getprime}(q),\; \mathtt{output}(pq)], f())\}.
\]

\subsection{Software Constant-Time-Verification Effort}

Focusing on semantics and compiler verification, we have not invested any effort to optimize or streamline the software proof flow.
Nevertheless, for all programs we have proven constant-time using any of our semantics, the time and code required to prove the actual constant-time property is a small fraction (around 10\%) of proving memory safety and defined behavior.
In particular, Coq's built-in proof-context tracking can automatically prove that an expression does not depend on a particular variable.
Most examples took less than one working day to implement, specify, and prove, except for Kyber message decoding, which took less than one workweek.

\section{Related Work}
A variety of projects have applied formal verification to different types of compilers for general-purpose languages, most notably CompCert~\cite{Leroy-backend} and CakeML~\cite{CakeML}.
Some projects have been more specialized to the domain of cryptography, where the idea of constant time originated.
Fiat Cryptography~\cite{FiatCrypto} produces (with a Coq-verified compiler) performance-critical inner loops that are straightline code without explicit memory access and thus trivially verified as constant time.
The libraries HACL$^*$~\cite{HACL} and EverCrypt~\cite{EverCrypt} based on the F$^*$ language~\cite{Fstar} tackle a broader range of cryptographic functionality, where constant time must be established explicitly.
There are no associated mechanized proofs about compilation, not even for functional correctness.

\citet{SimpleCtComp} initiated mechanized proof that compilers preserves constant time, first for a simple compiler inspired by Jasmin~\cite{Jasmin}.
We already discussed work~\cite{CtCompCert} applying similar techniques to most phases of CompCert.
In general, this past work did not apply to nondeterministic programs.

\citet{StructuredLeakage} introduce the idea of \emph{structured leakage}, where leakage traces are accumulated in purpose-specific data structures. 
They handle allocation of stack memory, one of our main running examples.
However---as with CompCert---their source language provides no way to observe memory-allocation addresses, sidestepping some of the central challenges that motivate our work.
In contrast, our approach permits proving a program constant-time even if it branches on the stack pointer.
We also handle other kinds of nondeterminism, most interestingly input-output. 

We are not aware of prior work involving specification of constant-time properties either (1) as single-copy properties or (2) using anything analogous to our ``predictors.''
Also, most prior work on verification of compiler security properties---including all works cited in this section---uses small-step semantics and thus (unlike our approach) supports nontermination without difficulty.

We did not consider fine-grained security policies, merely tagging each value ``private'' or ``public.''
\citet{Broberg} overview a broader class of security policies.

\paragraph{Constant Time With Nondeterminism}
\citet{Sison_Murray_19} prove for a realistic concurrent program that threads do not leak to other threads via shared memory or timing leaks.
Their work significantly differs from ours in that they work with security properties as hyperproperties rather than single-copy properties.
A significant part of their contribution is a method for proving a certain security hyperproperty by considering only two executions at a time (one source, one target) rather than four at a time (two source, two target).
We avoid relating imperative executions at all.

\citet{refiners} also handle security properties in the presence of concurrency, but they focus only on static analysis, not considering compiler specifications or proofs.
To do so, they parameterize their semantics over \emph{refiners}  to make it deterministic, just like we parameterize our semantics over oracles in \cref{fix_stack_swap}.
We used oracles to model stack allocation, and they used refiners to model scheduling (for concurrency), but it is the same idea.

\paragraph{Speculation-Aware Security Properties}
In general, cryptographic constant time is an insufficient security condition in the face of hardware optimizations like speculation, as demonstrated with Spectre attacks~\cite{Spectre}.
Tools based on programming languages and compilers have been suggested that help regain timing security in the face of speculation.
One example is Blade~\cite{Blade}, which employs a type system to track information flows and ensure sufficient placement of relatively expensive source-level mitigations.

Recently, both \citet{spec2} and \citet{spec1} have demonstrated how to verify compilers against speculation-aware semantics.
Both consider an adversary's actions to be a source of nondeterminism,
modeled by parameterizing the semantics over \emph{directives} (analogous to oracles and refiners), which ``model the ability of an adversary to influence program execution''~\cite{spec2}. 
These speculation-aware notions of timing security may also benefit from generalization to different kinds of nondeterminism, as we have explored for traditional constant time.

\section{Conclusion}

Characterizing timing security of interactive programs requires going beyond cryptographic constant time, even as the ideas behind that established notion remain relevant.
We introduced a series of approaches to specifying security of interactive programs, proving appropriate relations between approaches and showing how compiler verification can be adapted to show preservation of specifications in these different styles.
Ideas for building on our results include connecting security proofs to include end-to-end results with respect to specific processors, bringing in some notion of computational complexity to rule out unrealistic adversaries, and extending to proof of new phases that ingest security specifications and use them to drive additional optimizations.

\section*{Artifact-Availability Statement}
\addcontentsline{toc}{section}{Artifact-Availability Statement}
An artifact~\cite{artifact} for this paper was evaluated and is freely available.

\begin{acks}
  This research was supported by the \grantsponsor{GS4}{National Science Foundation}{} under grant~\grantnum{GS4}{CCF-2130671}.
  Any opinions, findings, and conclusions or recommendations expressed in this material are those of the authors and do not necessarily reflect the views of the National Science Foundation.
\end{acks}

\appendix
\section{Separating Nondeterminism from Leakage Breaks Causality}\label{naive_bad}
Consider defining functions \(B\) and \(L\) that, given a leakage trace \(k\), return the \(\branchnoarg\) events and \(\leaknoarg\) events respectively (filtering out other entries).
Consider the following definition: a program is constant-time if there exists a function \(f\), depending only on the program's public values, such that given any leakage trace \(k\) of the program, we have \(L(k) = f(B(k))\).
\par This definition likely comes across as reasonable; indeed, it is perhaps the simplest possible attempt at generalizing regular constant time to accomodate compiler-resolved nondeterminism.
However, it is a bad definition.
It is trying to get at the fact that the leakage trace of a constant-time function should be a function of the underlying oracle \(\mathcal{A}\) resolving the nondeterminism.
The problem, however, is that \(B(k)\) is not solely a function of \(\mathcal{A}\); in particular, \(B(k)\) may depend on private values.
\par As a demonstration, consider the following function, and consider \verb|x| to be a private argument.
\begin{verbatim}
countdown(x) { while (x--) { stackalloc 1 as y } }
\end{verbatim}
Now, clearly we do not want to say that \verb|countdown| is constant-time.
It is branching on the private input \verb|x|.
Yet, we claim that it satisfies the flawed definition of constant time.
\par The leakage trace of \verb|countdown| will be of the form
\[[\leak 1; \branch y_1; \leak 1; \branch y_2; \cdots;  \leak 1; \branch y_x] \app [\leak 0].\]
So, the required function \(f\) is just \(\lambda b. \; [\leak 1] * (\length b) \app [\leak 0]\).
Indeed, for any \(k\) that is a leakage trace of \verb|countdown|, we have \(L(k) = f(B(k))\).
Thus \verb|countdown| is constant-time according to the flawed definition.
\par We could frame the issue with the flawed definition as one of retrocausality.
When we just require the existence of the function \(f\) with \(L(k) = f(B(k))\), we allow an event \(e\) in the leakage events \(L(k)\) to depend arbitrarily on \(B(k)\); in particular, it can depend on the length of \(B(k)\), which in turn may depend on how the program executes even \emph{after} \(e\) is added to the leakage trace.
We conclude that the problem with separating the leakage events \(L(k)\) from the nondeterminism \(B(k)\) is that we lose any information about relative ordering between leakage events and nondeterministic events.
\par Finally, we remark that we could have gone similarly astray in defining the predicate \(\DownarrowA\).
Suppose that, in our definition of \(\Downarrow_\mathcal{A}\), we had chosen to encode \(\mathcal{A}\) not as a function \(\mathcal{A} : X \to Y\) but rather as a list of decisions \(\mathcal{A} \in Y^c\), where \(c\) is the number of oracle calls made during the program's execution.
When the length \(c\) of \(\mathcal{A}\) is not equal to the number of oracle calls made by a program \(p\), we say that \((p, m, \ell, t, k) \DownarrowA Q\) vacuously holds for every \(Q\).
\par If we had defined \(\DownarrowA\) in this way, then we would have been able to use the same hack to show that \verb|countdown| is ``constant-time'' according to \(\DownarrowA\).
That is, if we had defined \(\DownarrowA\) in this way, it would be true that
\[\exists f. \; \forall \mathcal{A}. \; \forall x. \; \forall m,\ell. \; (\verb|countdown(x)|, m, \ell[\verb|x| := x], [], []) \DownarrowA \{([], f(\mathcal{A}))\}.\]
The point is that the leakage \(f(\mathcal{A})\) appears to be independent of the value \(x\); but in fact it is not, since the length of \(\mathcal{A}\) depends on \(x\)!
Thus we see that correctly defining \(\DownarrowA\) is subject to the same subtlety as with our flawed definition of constant time at the beginning of this section.

\section{An Equivalence Result}\label{sec:equiv2}
Here is the equivalence result between \(\DownarrowA\) and \(\NdetDownarrowA\).
We originally presented this as \cref{equiv}.
\begin{theorem}\label{equiv2}
  For all \(p,m,\ell,Q\),
  \[\left[\forall \mathcal{A}. \; (p, m, \ell, [], []) \DownarrowA Q(\mathcal{A})\right] \iff \left[\forall \mathcal{A}. \; (p, m, \ell, [], []) \NdetDownarrowA Q(\mathcal{A})\right].\]
\end{theorem}
\par Before proving \cref{equiv2}, we note that, by expanding the definition of \(\NdetDownarrowA\) in the theorem, we obtain an immediate corollary.
\begin{corollary}\label{equiv_cor}
  For all \(p, m, \ell, Q\),
  \[\left[\forall \mathcal{A}. \; (p, m, \ell, [], []) \DownarrowA Q\right] \iff (p, m, \ell, [], []) \Downarrow Q.\]
\end{corollary}
\par We also observe a non-corollary.
Fix an \(\mathcal{A}\).
It is \emph{not} generally true that \((p, m, \ell, [], []) \DownarrowA Q \iff (p, m, \ell, [], []) \NdetDownarrowA Q\).
For instance, suppose \(p\) is a well-behaved terminating program satisfying \(Q\) when it executes compatibly with \(\mathcal{A}\), but there exists some other oracle \(\mathcal{B} \neq \mathcal{A}\) which causes \(p\) to loop or crash.
Then, since \(\Downarrow\) (and by extension, \(\NdetDownarrowA\)) requires \(p\) to terminate successfully for every possible nondeterministic behaviour, we cannot prove that \(p\) satisfies \emph{any} postcondition according to \(\NdetDownarrowA\), let alone \(Q\).

\par Now, we discuss the proof of \cref{equiv2}.
The leftward direction is easy.
\begin{lemma}\label{lemma1}
  For all \(p, m, \ell, Q, \mathcal{A}\),
  \[(p, m, \ell, [], []) \Downarrow Q \implies (p, m, \ell, [], []) \DownarrowA Q.\]
\end{lemma}
\begin{proof}
  The \(\Downarrow\) requires \(Q\) to be achieved regardless of stack-allocation addresses, whereas \(\Downarrow_\mathcal{A}\) only requires it for particular stack-allocation addresses returned by \(\mathcal{A}\).
\end{proof}
Note that, in the following lemma, we pretend that a postcondition only takes the leakage trace \(k\) as input,
a helpful simplification that we will make in the rest of this section.
The proof is completely analogous if we permit the postcondition to take the full final state \((m, \ell, t, k)\) as input.
\begin{lemma}\label{lemma2}
  For all \(p, m, \ell, Q, \mathcal{A}\),
  \[(p, m, \ell, [], []) \DownarrowA Q \implies (p, m, \ell, [], []) \DownarrowA \{k : k \sim \mathcal{A} \land Q(k)\}.\]
\end{lemma}
\begin{proof}
  Trivial, by induction (since \(\DownarrowA\) requires \(p\) to execute compatibly with \(\mathcal{A}\)).
\end{proof}

\begin{proof}[Proof of \cref{equiv2}(\(\impliedby\)).]
  Fix an \(\mathcal{A}\).
  By assumption,
 \[(p, m, \ell, [], []) \Downarrow \{k : k \sim \mathcal{A} \Rightarrow Q(\mathcal{A})(k)\}.\]
 By \cref{lemma1} then,
 \[(p, m, \ell, [], []) \DownarrowA \{k : k \sim \mathcal{A} \Rightarrow Q(\mathcal{A})(k)\}.\]
 By \cref{lemma2} then,
 \[(p, m, \ell, [], []) \DownarrowA \{k : k \sim \mathcal{A} \land (k \sim \mathcal{A} \Rightarrow Q(\mathcal{A})(k))\}.\]
  Finally, by weakening this last postcondition, we get \((p, m, \ell, [], []) \Downarrow_\mathcal{A} Q(\mathcal{A})\).
\end{proof}
So, the left implication of \cref{equiv2} is easy.
The right implication appears much harder.
A superficial way of capturing the difficulty is to ask, ``what should we do induction on?''
The only options are our infinitely many hypotheses of the form \((p, m, \ell, [], []) \Downarrow_\mathcal{A} Q(\mathcal{A})\).
None of these alone gives us enough information to even conclude that \((p, m, \ell, [], []) \Downarrow \{k : \mathsf{true}\}\); we need to take all of them together.
\par So, our intuition that the right implication of \cref{equiv2} is true appears to have little to do with the inductive structure of the semantics judgments.
Let us examine the intuition more carefully, in order to find a proof strategy.
Here is an informal sketch of the right implication.
\begin{proof}[Proof of \cref{equiv2} (\(\implies\)) (informal)]
  Fix an \(\mathcal{A}\).
  We want to show that \((p, m, \ell, [], []) \Downarrow \{k : k \sim \mathcal{A} \Rightarrow Q(\mathcal{A})(k)\}\).
  It ought to suffice to show that for any possible execution \(E\) of \((p, m, \ell, [], [])\) (where we think of an execution just as a possibly infinite list of states, and the set of possible executions is defined by \(\Downarrow\)), we have that (1) \(E\) terminates without crashing, and (2) if \(E\) was an execution with the compiler's nondeterministic choices matching \(\mathcal{A}\), then the final leakage \(k\) of \(E\) satisfies \(Q(\mathcal{A})(k)\).
  \par To prove that a given \(E\) terminates without crashing, we just note that the compiler's nondeterministic choices in \(E\) must be described by some oracle \(\mathcal{A}_E\).
  Then \(E\) is also a possible execution of \((p, m, \ell, [], [])\) according to \(\Downarrow_{\mathcal{A}_E}\), and we appeal to our hypothesis that \((p, m, \ell, [], []) \Downarrow_{\mathcal{A}_E} Q(\mathcal{A}_E)\) to see that \(E\) terminates without crashing.
  To prove (2), we note that if stack-allocation nondeterminism was described by \(\mathcal{A}\) in the execution of \(E\), then \(E\) is a possible execution of \((p, m, \ell, [], [])\) according to \(\Downarrow_\mathcal{A}\), and we appeal to our hypothesis that \((p, m, \ell, [], []) \Downarrow_\mathcal{A} Q(\mathcal{A})\) to see that the final leakage \(k\) of \(E\) must satisfy \(Q(\mathcal{A})(k)\).
\end{proof}
The handwavy part of that proof is the part where we assume that the proposition \((p, m, \ell, [], []) \Downarrow Q\) is equivalent to a statement of the form ``for any possible execution \(E\) of \((p, m, \ell, [], [])\), \ldots.''
So, our proof strategy is just to formalize this intuition.
We will show that \((p, m, \ell, [], []) \Downarrow Q\) is equivalent to a statement of that form (and do similarly for \(\Downarrow_\mathcal{A}\)), and then the remaining proof work is as easy as the informal proof given above.
\par We formalize this idea in terms of a small-step operational semantics, which otherwise is not needed in the formalization style that we adopted from Bedrock2.
This small-step semantics works with configurations of the form $(p, m, \ell, t, k)$, with the twist that we allow a slightly expanded grammar of programs $p$, so that (for example) we can leave markers for the ends of regions with local stack allocation (implicit with \verb|stackalloc| in source programs).
Here are two example rules, including just the most important parts of states.

\infrule
    {(e, (k, m, \ell)) \Downarrow (v, k')}
    {(x := e, (k, m, \ell)) \to (\mathsf{skip}, (k', m, \ell[x \hookleftarrow v]))}

\infrule
    {[a, a+n) \cap \mathsf{dom}~m = \varnothing{} 
        \andalso \mathsf{dom}(m') = [a, a+n)}
    {(\mathsf{stackalloc} \; n \; \mathsf{as} \; x \; \mathsf{in} \; p, (k, m, \ell)) \to (p; \mathsf{end\_stackalloc}(n, a), (k + \mathsf{allocate}(a), m \uplus m', \ell[x \hookleftarrow a]))}

\begin{definition}
  An \emph{execution} is an infinite sequence of optional states, a partial function of type $\mathbb N \rightharpoonup \mathsf{program} \times \mathsf{state}$.  An execution is \emph{possible} if, between every natural number $n$ and its successor $n+1$, either the successive mappings are compatible with small-step relation $\to$ (the program is not finished), $n$ is mapped to a stuck state and $n+1$ is not in the function's domain (the program just got stuck), $n$ is mapped to a state of the form $(\mathsf{skip}, s)$ and $n+1$ is not in the function's domain (the program just terminated successfully), or neither $n$ nor $n+1$ is in the domain (the program terminated or got stuck earlier).
\end{definition}

When an execution gets stuck, it is not always the fault of the program.
Thus we define a set $\mathcal{G}$ of states that are stuck for benign reasons.
In Bedrock2, there are two such reasons: either we have asked for IO, but there is no possible input meeting the IO specification; or we have asked to do some stack allocation, but there is no address available (out of memory).

\begin{definition}
  A small-step configuration satisfies a postcondition, written $(p, s) \models Q$, when either $(p, s) \in \mathcal{G}$ (benign stuckness) or $p = \mathsf{skip}$ and $Q(s)$ (execution finished in a state satisfying the postcondition).
  An execution satisfies a postcondition, written $E \models Q$, when there exists $i \in \mathbb N$ such that $E(i) \models Q$.
\end{definition}

\begin{lemma}\label{big_small}
  The following are equivalent:
  \begin{itemize}
  \item $(p, s) \Downarrow Q$.
  \item For any possible execution $E$ where $E(0) = (p, s)$, it follows that $E \models Q$.
  \end{itemize}
\end{lemma}

This lemma corresponds to the point in the informal proof of \cref{equiv2} where we wrote that \(p \Downarrow Q\) holds iff every possible execution beginning with \(p\) terminates without crashing, and the terminated state will satisfy the postcondition.
\par Now, to prove that a big-step judgment implies a small-step judgment, we just do a straightforward induction on the $\Downarrow$ judgment.
(This strategy is business as usual for big-step semantics, though we retain omnisemantics's peculiarity of supporting nondeterminism by ``running'' to exhaustive postconditions, not single outcomes.)
\par The small-to-big-step direction is harder, because it is not clear what to do induction on.
We need to define an unusual well-founded ordering on states.

\begin{definition}
  Define $\sqsubset$ as the transitive closure of the union of two basic relations:
  \begin{itemize}
  \item The small-step relation $\to$ with argument order reversed
  \item $\{(p1, (p1; p2)) \mid p1, p2 \in \mathsf{program}\}$
  \end{itemize}
\end{definition}
\begin{lemma}
  Relation $\sqsubset$ is well-founded when we restrict its arguments to just those configurations that have possible executions, each of which can be proved to satisfy at least one postcondition.
\end{lemma}
Now we can prove \cref{big_small} by well-founded induction on configurations w.r.t. $\sqsubset$.
This induction is helpful, because whenever we need to apply an inductive hypothesis in the proof of \cref{big_small}, it will either be to some further progressed state or to a subterm that is a prefix, the two basic cases in the definition of $\sqsubset$.
\par Once we have proved \cref{big_small}, the remaining formal proof of \cref{equiv2} (\(\implies\)) straightforwardly follows the informal proof sketch given further above.

\section{Leakage Trees as an Alternative to Predictors}\label{leakage_trees}
We found the predictors of \cref{predictors} to be inconvenient for source-program proofs and automation.
The aesthetic annoyance is just that it is not nice to write down or think about a predictor corresponding to a source program.
\par Worse, for automation purposes, consider proving a theorem of the following form (note that all of our predictor-style specifications in this paper are basically of this form).
\begin{verbatim}
exists p: predictor, ...
\end{verbatim}
It is hard to get Coq to infer what \verb|p| should look like.
Functions are not inductively defined, so one must either enter the correct value of \verb|p| manually or write a moderately clever script.
\par To solve this problem, \cref{leakage_trees_def} will present \emph{leakage trees}, which are something like ``inductively defined predictors,'' with an inductively defined version of the ``predicts'' predicate.
As an added bonus, leakage trees seem conceptually closer to leakage traces than predictors are.
\par \cref{more_pred} observes that leakage trees are slightly less expressive than predictors.
\cref{big_trees} presents a slightly modified version of leakage trees---``big trees''---along with a theorem stating that big trees are equally as expressive as predictors.
Finally, \cref{pred_vs_tree} discusses why trees are nicer to work with than predictors.

\subsection{Definition of Leakage Trees}\label{leakage_trees_def}
Without nondeterminism, we say that a program \(p\) is constant-time if it executes with a leakage trace that is a function of public inputs.
With nondeterminism, we say that \(p\) is predictor-constant-time if it executes with a predictor that is a function of public inputs.
\par The similarity between these two statements suggests that, in the nondeterministic case, a predictor actually plays the same role that a leakage trace plays in the deterministic case.
Inspired by this similarity, we will define \emph{leakage trees}: data structures that represent sets of leakage traces, just like predictors do, but which have the conceptual advantage of looking much more like leakage traces than predictors do.
\par Consider the following inductive definition, written in Coq.
The \verb|word| type is the machine-word type that we have been working with throughout this paper.
\begin{verbatim}
Inductive leakage_tree :=
| TreeLeaf
| TreeLeak (thing_to_leak : word) (rest_of_leakage : leakage_tree)
| TreeBranch (rest_of_leakage_given_branch_direction : word -> leakage_tree).
\end{verbatim}
Like a predictor, a leakage tree represents a certain set of leakage traces.
In particular, it represents the set of leakage traces that can be formed by taking a nonbacktracking path through the leakage tree, from the root to a leaf.
We formalize that definition as follows.
\begin{definition}\label{path}
  A leakage trace \(k\) is \emph{a path in} a leakage tree \(\mathcal{T}\) if \(k \in \mathcal{T}\), where \(\in\) is defined inductively by the following rules.
  \begin{mathpar}
  \inferrule[nil\_path]
     {\;}
     {[] \in \treeleaf}

  \inferrule[leak\_path]
     {k \in \mathcal{T}}
     {[\leak x] \app k \in \treeleak{x}{\mathcal{T}}}

  \inferrule[branch\_path]
     {k \in f(x)}
     {[\branch x] \app k \in \treebranch{f}}
  \end{mathpar}
\end{definition}
\begin{definition}
  Given a function \(\mathcal{T}\) taking IO traces to leakage trees, we say that \((p, m, l, [], [])\) \emph{executes with leakage tree} \(\mathcal{T}\) if \((p, m, l, [], []) \Downarrow \{(t, k) : k \in \mathcal{T}(t)\}\).
\end{definition}
The analogy between this definition and \cref{ex_with_pred} should be clear.
A leakage tree, like a predictor, is a natural way of answering the question ``what will the leakage trace of \((p, m, \ell, [], [])\) be?'' without having to assume the existence of an oracle $\mathcal{A}$.
We say that a program is \emph{tree-constant-time} if it executes with a leakage tree depending only on public values.

\subsection{There Are More Predictors than Trace Trees}\label{more_pred}
\par One might hope that a set of leakage traces can be specified by a predictor if and only if it can be specified by a leakage tree.
One direction is true.
\begin{theorem}\label{trees_are_predictors}
  For every leakage tree \(\mathcal{T}\), there exists a predictor \(\mathcal{P}\) such that for all leakage traces \(k\), we have \(k \in \mathcal{T} \iff k \in \mathcal{P}\).
\end{theorem}
\begin{proof}
  Suppose the tree \(\mathcal{T}\) is given.
  Given a leakage trace \(k\), to determine what \(\mathcal{P}(k)\) should be, all we have to do is walk along the tree \(\mathcal{T}\), starting from the root and then following the path defined by \(k\), until we exhaust \(k\).
  Then, wherever we end up tells us what event should come after \(k\): if we end at a leaf, we return \(\en\); if we end at a branch, we return \(\branchnoarg\); if we end at \(\treeleak x\), then we return \(\leak x\).
  \par In Coq, we can define \(\mathcal{P}\) as follows.
\begin{verbatim}
Fixpoint P (T : leakage_tree) (k : leakage) :=
  match T, k with
  | TreeLeaf, nil => End
  | TreeLeak x T', nil => Leak x
  | TreeBranch T', nil => CompNonDet
  | TreeLeak _ T', Leak _ :: k' => P T' k'
  | TreeBranch T', CompNonDet b :: k' => P (T' b) k'
  | _, _ => End (*in this case, k is not a prefix of a path in T,
                  so we can return anything*) 
  end.
\end{verbatim}  
\end{proof}
As a corollary, we conclude that tree-constant-time implies predictor-constant-time.
\par The converse to \cref{trees_are_predictors} is not true.
Predictors are able to represent some sets that leakage trees cannot represent.
\begin{example}\label{not_a_tree}
  Let \(\mathcal{P}\) be a predictor which, given a leakage trace \(k\), returns \(\en\) if the last element of \(k\) is \(\branch 0\) and returns \(\branchnoarg\) otherwise.
  There does not exist a leakage tree \(\mathcal{T}\) such that \(\forall k. \; k \in \mathcal{P} \iff k \in \mathcal{T}\).
\end{example}
\begin{proof}
  Assume we had such a \(\mathcal{T}\).
  Clearly \(\mathcal{T}\) is of the form \(\treebranch{f}\), where \(f(0) = \treeleaf\), and \(f(1)\) is such that \(\forall k. \; k \in \mathcal{P} \iff k \in f(1)\).
  \par Now let \(\mathcal{T}\) be minimal (under the subterm relation that Coq works with) with the property that \(\forall k. \; k \in \mathcal{P} \iff k \in \mathcal{T}\).
  But then \(\mathcal{T}\) is of the form \(\treebranch{f}\), where \(f(1)\) also has this property,
  contradicting the minimality of \(\mathcal{T}\).
\end{proof}
The issue demonstrated by this example is the only reason that some predictors do not have corresponding leakage trees.
Slightly more precisely, the only limitation is that every path in a leakage tree is finite.
Thus, as we show in the next section, there is a correspondence between predictors and potentially infinite rooted trees, where a trace is predicted by the predictor if and only if it is a finite root-to-leaf path in the tree.

\subsection{Big Trees Are Equivalent to Predictors}\label{big_trees}
Define a \emph{big tree} as follows.
\begin{verbatim}
CoInductive big_tree :=
| BigTreeLeaf
| BigTreeLeak (thing_to_leak : word) (rest_of_leakage : big_tree)
| BigTreeBranch (rest_of_leakage_given_branch_direction : word -> big_tree).
\end{verbatim}
Note that the definition of \verb|big_tree| is exactly the same as the definition of \verb|leakage_tree|, the only difference being that we are taking the coinductive interpretation of the constructors---that is, the word \verb|Inductive| has been replaced by \verb|CoInductive|.
\begin{definition}\label{big_path}
  A leakage trace \(k\) is \emph{a path in} a big tree \(\mathcal{B}\) if \(k \in \mathcal{B}\), where \(\in\) is defined inductively by the following rules.
  \begin{mathpar}
  \inferrule[nil\_path]
     {\;}
     {[] \in \bigtreeleaf}

  \inferrule[leak\_path]
     {k \in \mathcal{B}}
     {[\leak x] \app k \in \bigtreeleak{x}{\mathcal{B}}}

  \inferrule[branch\_path]
     {k \in f(x)}
     {[\branch x] \app k \in \bigtreebranch{f}}
  \end{mathpar}
\end{definition}
Note that \cref{big_path} is exactly the same as \cref{path}, the only difference being that trees have been replaced with big trees.
We can then define \emph{big-tree constant time} analogously to tree constant time.

\begin{theorem}\label{big_trees_exactly_predictors}
  For every big tree \(\mathcal{B}\), there exists a predictor \(\mathcal{P}\) such that for all leakage traces \(k\), we have \(k \in \mathcal{B} \iff k \in \mathcal{P}\).
  Conversely: for every predictor \(\mathcal{P}\), there exists a big tree \(\mathcal{B}\) such that for all \(k\), we have \(k \in \mathcal{B} \iff k \in \mathcal{P}\).
\end{theorem}
\begin{proof}
  Straightforward; similar to the proof of \cref{trees_are_predictors}.
\end{proof}
By the theorem, big-tree constant time is the same thing as predictor constant time.
Further, everything that is done with predictors in this paper could just as well have been done with big trees instead.
\par For an exploration of data structures similar to leakage trees and big trees, see \cite{itrees}.

\subsection{Practical Considerations: Predictors vs. Leakage Trees}\label{pred_vs_tree}
We have experimented with writing source-program and compiler proofs using leakage trees.
We did not feel the need to experiment with big trees, since---despite the fact that leakage trees are less expressive than big trees---we did not encounter any example programs that are big-tree constant time but not tree constant time.
Indeed, although \cref{not_a_tree} shows that big-tree constant time and tree constant time are not equivalent in general, we suspect (but have not proved) that they are equivalent for terminating programs.
\par Since we have not experimented with big trees, the observations discussed in this subsection will be just about leakage trees rather than big trees.
However, we expect the same considerations to apply to big trees.
\par Leakage trees seem nicer to work with and think about than predictors are.
The inductive structure of leakage trees, and the inductive structure of the \(k \in \mathcal{T}\) relation, make them more suitable for proof automation in Coq than predictors are.
In our experience of writing program proofs, showing that there exists a leakage tree with which a small program executes usually amounts to symbolically executing the program to see what its leakage trace \(k\) is, obtaining a goal of the form \(k \in\,\, ?\mathcal{T}\), and then writing \verb|repeat econstructor| to let Coq fill in the value of \(\mathcal{T}\).
\par A compiler specification using big trees or leakage trees can be written just like a compiler specification using predictors; to say what it means for a compiler to admit tree-transformation functions or big-tree-transformation functions, just modify \cref{ctp3} by replacing the predictor \(\mathcal{P}\) with a leakage tree \(\mathcal{T}\) or a big tree \(\mathcal{B}\).
We have experimented with proving that a compiler admits tree-transformation functions.
It is similar to proving that a compiler admits predictor-transformation functions.

\subsection{A Predictor and Leakage Tree for \texttt{stack\_swap}}\label{ss_pred_tree}
Recall that in \cref{pppct} we declined to write down a predictor for \verb|stack_swap|, saying that it would be awkward and unenlightening.
Nonetheless, here it is.
\begin{verbatim}
fun k =>
match k with
| [] => CompNonDet
| [CompNonDet x] => Leak x
| [CompNonDet x; Leak _] => Leak (x + 1)
| [CompNonDet x; Leak _; Leak _] => Leak x
| [CompNonDet x; Leak _; Leak _; Leak _] => Leak (x + 1)
| _ => end
end
\end{verbatim}
For comparison, here is a leakage tree for \verb|stack_swap|.
\begin{verbatim}
TreeBranch (fun x => TreeLeak x (TreeLeak (x + 1)
  (TreeLeak x (TreeLeak (x + 1) TreeLeaf))))
\end{verbatim}

\section{How to Make Predictor-Based Specifications Modular for Program Verification}\label{sec:nonmodular}
Suppose we have a function \verb|f2| that calls a function \verb|f1|, and we have proved a specification of \verb|f1| saying that it executes with a predictor that depends only on public values.
We would like to use this specification of \verb|f1| to help prove that \verb|f2| executes with a predictor depending only on public values.
But as we will show, the specification of \verb|f1| may not be helpful for proving the specification of \verb|f2|; in this sense, program proofs using predictors appear to be unfortunately nonmodular.
\par For example, consider these functions, where \verb|a| is some global constant.
\begin{verbatim}
  f1 () { stackalloc 4 as x; *a = x }
  f2 () { f1(); x = *a; if (x) {} }
\end{verbatim}
Suppose we have proved that \verb|f1| executes with predictor \(\mathcal{P}_1\).
Now we would like to use this fact to prove that \verb|f2| executes with some fixed predictor.
But clearly the fact that \verb|f1| executes with predictor \(\mathcal{P}_1\) is not a sufficient specification.
Instead, we need a specification saying ``\verb|f1| executes with predictor \(\mathcal{P}_1\), \emph{and} the value that \verb|f1| stores at address \verb|a| does not depend on private values.''
\par But how do we formalize that ``the value that \verb|f1| wrote at address \verb|a| does not depend on private values''?
It is not immediately obvious how to do so without resorting to one of the predicates \(\DownarrowA\) or \(\NdetDownarrowA\).
The trick is to take advantage of the fact that \verb|f1| is constant-time, and thus there should be nothing private in its leakage trace.
So, intuitively, to specify that the value written at address \verb|a| depends only on public values, it should suffice to say that the value is a function of the leakage trace of \verb|f1|.
Formally, we write the specification of \verb|f1| as follows (where \(a\) is the value stored in the global constant \verb|a|).
\[\exists \mathcal{P}_1, f. \; \forall m, \ell. \; (\verb|f1()|, m, \ell, [], []) \Downarrow \{(m', \ell, [], k) : k \in \mathcal{P}_1 \land m' @ a = f(k)\}.\]
To construct \(f\) satisfying the specification, we can note that the stack-allocation address appears in the leakage trace \(k\), and hence the function \(f\) can extract the stack-allocation address from \(k\).
\par Now, we demonstrate that this specification of \verb|f1| can be used to prove the following specification of \verb|f2|.
\[\exists \mathcal{P}. \; \forall m, \ell. \; (\verb|f2()|, m, \ell, [], []) \Downarrow \{(m', \ell, [], k) : k \in \mathcal{P}\}.\]
To do so, we use the following general lemma.
\begin{lemma}
  Let \(\mathcal{P}_1\) be a predictor, and let \(\mathcal{P}_2\) be a function taking leakage traces to predictors.
  Then, there exists a predictor \(\mathcal{P}\), which we denote by \(\mathcal{P}_1 \app \mathcal{P}_2\), such that
  \[\forall k_1, k_2. \;\;\; k_1 \in \mathcal{P}_1 \land k_2 \in \mathcal{P}_2(k_1) \implies (k_1 \app k_2) \in \mathcal{P}.\]
\end{lemma}
We will not write a proof of the lemma here, but we note that it is not hard.
And, like all of our other lemmas and theorems, it is formalized in Coq.
\par Now, we observe that the leakage of \verb|f2| is of the form \(k_1 \app k_2\), where \(k_1\) is the leakage of \verb|f1|, and \(k_2\) is the leakage of everything coming after the call to \verb|f1|.
Further, the specification of \verb|f1| equips us with \(\mathcal{P}_1\) such that \(k_1 \in \mathcal{P}_1\) and \(f\) such that \(f(k_1)\) is the value stored at \verb|a|.
Also, it is clear how to define a function \(\mathcal{P}_2\) which takes as input the value stored at \verb|a| and returns a predictor for \(k_2\).
So, by the lemma, we obtain a predictor \(\mathcal{P}_1 \app (\mathcal{P}_2 \circ f)\) for the whole leakage trace of \verb|f2|.
\par The method presented here generalizes to the case where \verb|f1| outputs any set of values that are then observed by \verb|f2|.
The general principle is that, if we have a predictor for the leakage \(k_1\) of \verb|f1|, and we have a function taking \(k_1\) to a predictor \(\mathcal{P}_2\) for the leakage \(k_2\) that comes after \verb|f1|, then we can concatenate them to get a predictor for \verb|f2|.
\par The concern considered in this section (and the corresponding solution) applies to leakage-tree-based specifications just as it applies to predictor-based specifications.

\section{Predictor-Based Compiler-Correctness Theorems for Reordering Optimizations}\label{app:compile-predictor}

\subsection{A Reasonable Compiler Optimization that Violates \cref{ctp5}}\label{sec:optimization}
In this section, we present a simple example of a compiler reordering optimization that is intuitively constant time-preserving but does not actually admit oracle-transformation functions.
Because the correctness of reordering memory allocations gets subtle very quickly (see, e.g., address guessing \cite{twinsem}), it would be a distraction to work with our running example of \verb|stackalloc| here.
Instead, we introduce a new language construct, \verb|random as x|, that picks a random number and binds it to the local variable \verb|x|.
The random number is picked in the same way as \verb|stackalloc| addresses are: by applying the oracle \(\mathcal{A}\) to the previous leakage trace.
\par Note that the examples and proofs involving the \verb|random| construct are about the only part of this paper that is purely on paper with no Coq proofs.  Bedrock2 does not have the \verb|random| construct.
\par Define the programs \(p\) and \(p'\) as follows.
\begin{alltt}
\(p\,\,:=\) \{ random as x; z = *w; print(x) \}
\(p':=\) \{ z = *w; random as x; print(x) \}
\end{alltt}
Arguably, a constant time-preserving compiler should be permitted to transform \(p\) to \(p'\).
Yet, consider the following fact.
\begin{example}
  Let \(C\) be a compiler pass with \(C(p) = p'\).
  Then \(C\) cannot satisfy \cref{ctp5}.
\end{example}
\begin{proof}
  Suppose that \(C\) did satisfy \cref{ctp5}.
  Then let \(\mathcal{A}_p\) be as in \cref{ctp5}.
  Let \(a_1, a_2\) be two distinct memory addresses.
  Let \(\mathcal{A}\) be a low-level oracle with \(\mathcal{A}([\leak a_1]) \neq \mathcal{A}([\leak a_2])\).
  Let \(\mathcal{B} := \mathcal{A}_p(\mathcal{A})\).
  Let \(m\) be some memory state where reads from \(a_1, a_2\) are allowed.
  Let \(\ell\) be any local-variables state.
  \par For \(i = 1, 2\) it is clear from looking at \(p\) that
  \[\forall \mathcal{C}. \; (p, m, \ell[\mathtt{w} := a_i], [], []) \DownarrowC \{([\verb|out | \mathcal{C}([])], k) : k \textrm{ arbitrary}\}.\]
  Applying \cref{ctp5}, we obtain
  \[\forall \mathcal{C}. \; (p', m, \ell[w := a_i], [], []) \DownarrowC \{([\verb|out | \mathcal{A}_p(\mathcal{C})([])], k) : k \textrm{ arbitrary}\}.\]
  Setting \(\mathcal{C} := \mathcal{A}\) in the proposition above, we obtain
  \begin{equation}\label{eqone}(p', m, \ell[\mathtt{w} := a_i], [], []) \DownarrowA \{([\verb|out | \mathcal{B}([])], k) : k \textrm{ arbitrary}\}.\end{equation}
  However, it is clear from looking at \(p'\) that
  \begin{equation}\label{eqtwo}(p', m, \ell[w := a_i], [], []) \DownarrowA \{([\verb|out | \mathcal{A}([\leak a_i])], k) : k \textrm{ arbitrary}\}.\end{equation}
  Now, setting \(i = 1\) in \cref{eqone} and \cref{eqtwo}, we infer that \(\mathcal{A}([\leak a_1]) = \mathcal{B}([])\).
  Similarly setting \(i = 2\), we get \(\mathcal{A}([\leak a_2]) = \mathcal{B}([])\).
  By transitivity, \(\mathcal{A}([\leak a_1]) = \mathcal{A}([\leak a_2])\),
  contradicting our definition of \(\mathcal{A}\).
\end{proof}
This proof exploited exactly the issue discussed in \cref{predictors_good}.
The compiler moved a nondeterministic event past a deterministic event, and thus it violated the assumption of the source-language semantics that the value of \verb|x| depends only on events occurring (in the source program) before the allocation of \verb|x|.
Note that the proof used only the existence of \(\mathcal{A}_p\) and not the existence of \(\gamma_p\); thus we have shown that \(C\) fails to preserve source-language semantics according to \(\DownarrowA\), even without considering how it transforms leakage traces.
\par Two objections may come to mind when viewing this example.
First, perhaps the issue only arises because we unnecessarily allow random-number-picking to depend on loads. 
But of course we could have a more complicated example, where the nondeterministic event is nondeterministic evaluation order for an expression, or nondeterministic memory allocation, etc.; and some arbitrarily complicated subcomputation that should be allowed to affect the outcome of the nondeterministic event (e.g, a loop with a variable number of executions) takes the place of the load.
\par Second, it can be observed that it is possible in principle to have an optimization pass that fails to satisfy \cref{ctp5}, when nevertheless the composition of all the phases satisfies \cref{ctp5}.
However, this constraint applies regardless of the target language---even when the target program is deterministic.
It significantly constrains how the target program chooses outcomes of events that are nondeterministic in the source language.
\par We therefore conclude that \cref{ctp5} imposes a real practical limitation.

\subsection{Compilers Can Perform Reordering Optimizations and Still Admit Predictor-Transformation Functions}
In this subsection we demonstrate that, indeed, a compiler can perform the program transformation of \cref{sec:optimization} and nontheless admit predictor-transformation functions.
\par More concretely: let \(C\) be a compiler pass with \(C(p) = p'\) (with \(p, p'\) as in \cref{sec:optimization}).
We will exhibit the predictor-transformation function \(\gamma_p\) as in \cref{ctp3}.
\par Given \(\mathcal{P}\), define \(\mathcal{P}' := \gamma_p(\mathcal{P})\) as follows.
Recall what \cref{ctp3} requires of us: if \(k\) is a possible trace of \(p\) such that \(k \in \mathcal{P}\), then the corresonding trace \(k'\) of \(p'\) should satisfy \(k' \in \mathcal{P}'\).
When the leakage of the source program is of the form \([\branch x; \leak w]\), the corresponding leakage of the target program will be \([\leak w; \branch x]\).
So, we set \(\mathcal{P}'([]) := \mathcal{P}([\branch 0])\).
Note that we can just assume that the branch taken is \(0\) because the value \(w\) that is leaked is \emph{independent of the branch \(x\)}.
\par The rest of defining \(\mathcal{P}'\) is straightforward.
For any \(w\), we set \(\mathcal{P}'([\leak w]) := \branch\), and for any \(w,x\) we set \(\mathcal{P}'([\leak w; \branch x]) := \en\).
On other inputs it does not matter what value \(\mathcal{P}'\) takes.
Now we have defined \(\mathcal{P}' = \gamma_p(\mathcal{P})\), and clearly \(\gamma_p\) satisfies the constraint of \cref{ctp3}.

\bibliographystyle{ACM-Reference-Format}
\bibliography{paper}


\begin{thebibliography}{29}


\ifx \showCODEN    \undefined \def \showCODEN     #1{\unskip}     \fi
\ifx \showDOI      \undefined \def \showDOI       #1{#1}\fi
\ifx \showISBNx    \undefined \def \showISBNx     #1{\unskip}     \fi
\ifx \showISBNxiii \undefined \def \showISBNxiii  #1{\unskip}     \fi
\ifx \showISSN     \undefined \def \showISSN      #1{\unskip}     \fi
\ifx \showLCCN     \undefined \def \showLCCN      #1{\unskip}     \fi
\ifx \shownote     \undefined \def \shownote      #1{#1}          \fi
\ifx \showarticletitle \undefined \def \showarticletitle #1{#1}   \fi
\ifx \showURL      \undefined \def \showURL       {\relax}        \fi
\providecommand\bibfield[2]{#2}
\providecommand\bibinfo[2]{#2}
\providecommand\natexlab[1]{#1}
\providecommand\showeprint[2][]{arXiv:#2}

\bibitem[Al~Fardan and Paterson(2013)]%
        {Lucky13}
\bibfield{author}{\bibinfo{person}{Nadhem~J. Al~Fardan} {and}
  \bibinfo{person}{Kenneth~G. Paterson}.} \bibinfo{year}{2013}\natexlab{}.
\newblock \showarticletitle{Lucky Thirteen: Breaking the TLS and DTLS Record
  Protocols}. In \bibinfo{booktitle}{\emph{Proceedings of the 2013 IEEE
  Symposium on Security and Privacy}} \emph{(\bibinfo{series}{SP '13})}.
  \bibinfo{publisher}{IEEE Computer Society}, \bibinfo{address}{USA},
  \bibinfo{pages}{526–540}.
\newblock
\showISBNx{9780769549774}
\urldef\tempurl%
\url{https://doi.org/10.1109/SP.2013.42}
\showDOI{\tempurl}


\bibitem[Almeida et~al\mbox{.}(2017)]%
        {Jasmin}
\bibfield{author}{\bibinfo{person}{Jos\'{e}~Bacelar Almeida},
  \bibinfo{person}{Manuel Barbosa}, \bibinfo{person}{Gilles Barthe},
  \bibinfo{person}{Arthur Blot}, \bibinfo{person}{Benjamin Gr\'{e}goire},
  \bibinfo{person}{Vincent Laporte}, \bibinfo{person}{Tiago Oliveira},
  \bibinfo{person}{Hugo Pacheco}, \bibinfo{person}{Benedikt Schmidt}, {and}
  \bibinfo{person}{Pierre-Yves Strub}.} \bibinfo{year}{2017}\natexlab{}.
\newblock \showarticletitle{Jasmin: High-Assurance and High-Speed
  Cryptography}. In \bibinfo{booktitle}{\emph{Proceedings of the 2017 ACM
  SIGSAC Conference on Computer and Communications Security}} (Dallas, Texas,
  USA) \emph{(\bibinfo{series}{CCS '17})}. \bibinfo{publisher}{Association for
  Computing Machinery}, \bibinfo{address}{New York, NY, USA},
  \bibinfo{pages}{1807–1823}.
\newblock
\showISBNx{9781450349468}
\urldef\tempurl%
\url{https://doi.org/10.1145/3133956.3134078}
\showDOI{\tempurl}


\bibitem[Arranz~Olmos et~al\mbox{.}(2025)]%
        {spec2}
\bibfield{author}{\bibinfo{person}{Santiago Arranz~Olmos},
  \bibinfo{person}{Gilles Barthe}, \bibinfo{person}{Lionel Blatter},
  \bibinfo{person}{Benjamin Gr\'{e}goire}, {and} \bibinfo{person}{Vincent
  Laporte}.} \bibinfo{year}{2025}\natexlab{}.
\newblock \showarticletitle{Preservation of Speculative Constant-Time by
  Compilation}.
\newblock \bibinfo{journal}{\emph{Proc. ACM Program. Lang.}}
  \bibinfo{volume}{9}, \bibinfo{number}{POPL}, Article \bibinfo{articleno}{44}
  (\bibinfo{date}{Jan.} \bibinfo{year}{2025}), \bibinfo{numpages}{33}~pages.
\newblock
\urldef\tempurl%
\url{https://doi.org/10.1145/3704880}
\showDOI{\tempurl}


\bibitem[Barthe et~al\mbox{.}(2019)]%
        {CtCompCert}
\bibfield{author}{\bibinfo{person}{Gilles Barthe}, \bibinfo{person}{Sandrine
  Blazy}, \bibinfo{person}{Benjamin Gr\'{e}goire}, \bibinfo{person}{R\'{e}mi
  Hutin}, \bibinfo{person}{Vincent Laporte}, \bibinfo{person}{David Pichardie},
  {and} \bibinfo{person}{Alix Trieu}.} \bibinfo{year}{2019}\natexlab{}.
\newblock \showarticletitle{Formal verification of a constant-time preserving C
  compiler}.
\newblock \bibinfo{journal}{\emph{Proc. ACM Program. Lang.}}
  \bibinfo{volume}{4}, \bibinfo{number}{POPL}, Article \bibinfo{articleno}{7}
  (\bibinfo{date}{Dec} \bibinfo{year}{2019}), \bibinfo{numpages}{30}~pages.
\newblock
\urldef\tempurl%
\url{https://doi.org/10.1145/3371075}
\showDOI{\tempurl}


\bibitem[Barthe et~al\mbox{.}(2021)]%
        {StructuredLeakage}
\bibfield{author}{\bibinfo{person}{Gilles Barthe}, \bibinfo{person}{Benjamin
  Gr\'{e}goire}, \bibinfo{person}{Vincent Laporte}, {and}
  \bibinfo{person}{Swarn Priya}.} \bibinfo{year}{2021}\natexlab{}.
\newblock \showarticletitle{Structured Leakage and Applications to
  Cryptographic Constant-Time and Cost}. In
  \bibinfo{booktitle}{\emph{Proceedings of the 2021 ACM SIGSAC Conference on
  Computer and Communications Security}} (Virtual Event, Republic of Korea)
  \emph{(\bibinfo{series}{CCS '21})}. \bibinfo{publisher}{Association for
  Computing Machinery}, \bibinfo{address}{New York, NY, USA},
  \bibinfo{pages}{462–476}.
\newblock
\showISBNx{9781450384544}
\urldef\tempurl%
\url{https://doi.org/10.1145/3460120.3484761}
\showDOI{\tempurl}


\bibitem[Barthe et~al\mbox{.}(2018)]%
        {SimpleCtComp}
\bibfield{author}{\bibinfo{person}{Gilles Barthe}, \bibinfo{person}{Benjamin
  Grégoire}, {and} \bibinfo{person}{Vincent Laporte}.}
  \bibinfo{year}{2018}\natexlab{}.
\newblock \showarticletitle{Secure Compilation of Side-Channel Countermeasures:
  The Case of Cryptographic “Constant-Time”}. In
  \bibinfo{booktitle}{\emph{2018 IEEE 31st Computer Security Foundations
  Symposium (CSF)}}. \bibinfo{pages}{328--343}.
\newblock
\urldef\tempurl%
\url{https://doi.org/10.1109/CSF.2018.00031}
\showDOI{\tempurl}


\bibitem[Bernstein et~al\mbox{.}(2024)]%
        {KyberSlash}
\bibfield{author}{\bibinfo{person}{Daniel~J. Bernstein},
  \bibinfo{person}{Karthikeyan Bhargavan}, \bibinfo{person}{Shivam Bhasin},
  \bibinfo{person}{Anupam Chattopadhyay}, \bibinfo{person}{Tee~Kiah Chia},
  \bibinfo{person}{Matthias~J. Kannwischer}, \bibinfo{person}{Franziskus
  Kiefer}, \bibinfo{person}{Thales Paiva}, \bibinfo{person}{Prasanna Ravi},
  {and} \bibinfo{person}{Goutam Tamvada}.} \bibinfo{year}{2024}\natexlab{}.
\newblock \bibinfo{title}{{KyberSlash}: Exploiting secret-dependent division
  timings in Kyber implementations}.
\newblock \bibinfo{howpublished}{Cryptology ePrint Archive, Paper 2024/1049}.
\newblock
\urldef\tempurl%
\url{https://eprint.iacr.org/2024/1049}
\showURL{%
\tempurl}
\newblock
\shownote{\url{https://eprint.iacr.org/2024/1049}}.


\bibitem[Bourgeat et~al\mbox{.}(2023)]%
        {riscvsem}
\bibfield{author}{\bibinfo{person}{Thomas Bourgeat}, \bibinfo{person}{Ian
  Clester}, \bibinfo{person}{Andres Erbsen}, \bibinfo{person}{Samuel Gruetter},
  \bibinfo{person}{Pratap Singh}, \bibinfo{person}{Andy Wright}, {and}
  \bibinfo{person}{Adam Chlipala}.} \bibinfo{year}{2023}\natexlab{}.
\newblock \showarticletitle{Flexible Instruction-Set Semantics via Abstract
  Monads (Experience Report)}.
\newblock \bibinfo{journal}{\emph{Proc. ACM Program. Lang.}}
  \bibinfo{volume}{7}, \bibinfo{number}{ICFP}, Article \bibinfo{articleno}{192}
  (\bibinfo{date}{aug} \bibinfo{year}{2023}), \bibinfo{numpages}{17}~pages.
\newblock
\urldef\tempurl%
\url{https://doi.org/10.1145/3607833}
\showDOI{\tempurl}


\bibitem[Broberg et~al\mbox{.}(2015)]%
        {Broberg}
\bibfield{author}{\bibinfo{person}{Niklas Broberg}, \bibinfo{person}{Bart van
  Delft}, {and} \bibinfo{person}{David Sands}.}
  \bibinfo{year}{2015}\natexlab{}.
\newblock \showarticletitle{The Anatomy and Facets of Dynamic Policies}. In
  \bibinfo{booktitle}{\emph{Proceedings of the 2015 IEEE 28th Computer Security
  Foundations Symposium}} \emph{(\bibinfo{series}{CSF '15})}.
  \bibinfo{publisher}{IEEE Computer Society}, \bibinfo{address}{USA},
  \bibinfo{pages}{122–136}.
\newblock
\showISBNx{9781467375382}
\urldef\tempurl%
\url{https://doi.org/10.1109/CSF.2015.16}
\showDOI{\tempurl}


\bibitem[Chargu\'{e}raud et~al\mbox{.}(2023)]%
        {Omnisemantics}
\bibfield{author}{\bibinfo{person}{Arthur Chargu\'{e}raud},
  \bibinfo{person}{Adam Chlipala}, \bibinfo{person}{Andres Erbsen}, {and}
  \bibinfo{person}{Samuel Gruetter}.} \bibinfo{year}{2023}\natexlab{}.
\newblock \showarticletitle{Omnisemantics: Smooth Handling of Nondeterminism}.
\newblock \bibinfo{journal}{\emph{ACM Trans. Program. Lang. Syst.}}
  \bibinfo{volume}{45}, \bibinfo{number}{1}, Article \bibinfo{articleno}{5}
  (\bibinfo{date}{mar} \bibinfo{year}{2023}), \bibinfo{numpages}{43}~pages.
\newblock
\showISSN{0164-0925}
\urldef\tempurl%
\url{https://doi.org/10.1145/3579834}
\showDOI{\tempurl}


\bibitem[Conoly et~al\mbox{.}(2025)]%
        {artifact}
\bibfield{author}{\bibinfo{person}{Owen Conoly}, \bibinfo{person}{Andres
  Erbsen}, {and} \bibinfo{person}{Adam Chlipala}.}
  \bibinfo{year}{2025}\natexlab{}.
\newblock \bibinfo{booktitle}{\emph{Machine-Checked Proof Artifact for
  Integrated Proofs of Cryptographic Constant Time for Nondeterministic
  Programs and Compilers}}.
\newblock
\urldef\tempurl%
\url{https://doi.org/10.5281/zenodo.15043688}
\showDOI{\tempurl}


\bibitem[Erbsen et~al\mbox{.}(2021)]%
        {lightbulb}
\bibfield{author}{\bibinfo{person}{Andres Erbsen}, \bibinfo{person}{Samuel
  Gruetter}, \bibinfo{person}{Joonwon Choi}, \bibinfo{person}{Clark Wood},
  {and} \bibinfo{person}{Adam Chlipala}.} \bibinfo{year}{2021}\natexlab{}.
\newblock \showarticletitle{Integration Verification across Software and
  Hardware for a Simple Embedded System}. In
  \bibinfo{booktitle}{\emph{Proceedings of the 42nd ACM SIGPLAN International
  Conference on Programming Language Design and Implementation}}
  \emph{(\bibinfo{series}{PLDI 2021})}. \bibinfo{publisher}{Association for
  Computing Machinery}, \bibinfo{pages}{604–619}.
\newblock
\showISBNx{9781450383912}
\urldef\tempurl%
\url{https://doi.org/10.1145/3453483.3454065}
\showDOI{\tempurl}


\bibitem[Erbsen et~al\mbox{.}(2019)]%
        {FiatCrypto}
\bibfield{author}{\bibinfo{person}{Andres Erbsen}, \bibinfo{person}{Jade
  Philipoom}, \bibinfo{person}{Jason Gross}, \bibinfo{person}{Robert Sloan},
  {and} \bibinfo{person}{Adam Chlipala}.} \bibinfo{year}{2019}\natexlab{}.
\newblock \showarticletitle{Simple High-Level Code for Cryptographic Arithmetic
  - With Proofs, Without Compromises}.
\newblock \bibinfo{journal}{\emph{2019 IEEE Symposium on Security and Privacy
  (SP)}} \bibinfo{volume}{54}, \bibinfo{number}{1} (\bibinfo{date}{May}
  \bibinfo{year}{2019}), \bibinfo{pages}{1202--1219}.
\newblock
\urldef\tempurl%
\url{https://doi.org/10.1109/sp.2019.00005}
\showDOI{\tempurl}


\bibitem[Erbsen et~al\mbox{.}(2024)]%
        {garagedoor}
\bibfield{author}{\bibinfo{person}{Andres Erbsen}, \bibinfo{person}{Jade
  Philipoom}, \bibinfo{person}{Dustin Jamner}, \bibinfo{person}{Ashley Lin},
  \bibinfo{person}{Samuel Gruetter}, \bibinfo{person}{Cl\'{e}ment Pit-Claudel},
  {and} \bibinfo{person}{Adam Chlipala}.} \bibinfo{year}{2024}\natexlab{}.
\newblock \showarticletitle{Foundational Integration Verification of a
  Cryptographic Server}.
\newblock \bibinfo{journal}{\emph{Proc. ACM Program. Lang.}}
  \bibinfo{volume}{8}, \bibinfo{number}{PLDI}, Article \bibinfo{articleno}{216}
  (\bibinfo{date}{jun} \bibinfo{year}{2024}), \bibinfo{numpages}{26}~pages.
\newblock
\urldef\tempurl%
\url{https://doi.org/10.1145/3656446}
\showDOI{\tempurl}


\bibitem[Gruetter et~al\mbox{.}(2024)]%
        {live}
\bibfield{author}{\bibinfo{person}{Samuel Gruetter}, \bibinfo{person}{Viktor
  Fukala}, {and} \bibinfo{person}{Adam Chlipala}.}
  \bibinfo{year}{2024}\natexlab{}.
\newblock \showarticletitle{Live Verification in an Interactive Proof
  Assistant}.
\newblock \bibinfo{journal}{\emph{Proc. ACM Program. Lang.}}
  \bibinfo{volume}{8}, \bibinfo{number}{PLDI}, Article \bibinfo{articleno}{209}
  (\bibinfo{date}{June} \bibinfo{year}{2024}), \bibinfo{numpages}{24}~pages.
\newblock
\urldef\tempurl%
\url{https://doi.org/10.1145/3656439}
\showDOI{\tempurl}


\bibitem[Kernighan and Ritchie(1988)]%
        {KRC88}
\bibfield{author}{\bibinfo{person}{Brian~W. Kernighan} {and}
  \bibinfo{person}{Dennis~M. Ritchie}.} \bibinfo{year}{1988}\natexlab{}.
\newblock \bibinfo{booktitle}{\emph{The C Programming Language}
  (\bibinfo{edition}{2} ed.)}.
\newblock \bibinfo{publisher}{Prentice Hall}.
\newblock
\urldef\tempurl%
\url{/bib/kernighan/Kernighan1988/the_c_programming_language_ritchie_kernighan.pdf}
\showURL{%
\tempurl}


\bibitem[Kocher et~al\mbox{.}(2019)]%
        {Spectre}
\bibfield{author}{\bibinfo{person}{Paul Kocher}, \bibinfo{person}{Jann Horn},
  \bibinfo{person}{Anders Fogh}, \bibinfo{person}{Daniel Genkin},
  \bibinfo{person}{Daniel Gruss}, \bibinfo{person}{Werner Haas},
  \bibinfo{person}{Mike Hamburg}, \bibinfo{person}{Moritz Lipp},
  \bibinfo{person}{Stefan Mangard}, \bibinfo{person}{Thomas Prescher},
  \bibinfo{person}{Michael Schwarz}, {and} \bibinfo{person}{Yuval Yarom}.}
  \bibinfo{year}{2019}\natexlab{}.
\newblock \showarticletitle{Spectre Attacks: Exploiting Speculative Execution}.
  In \bibinfo{booktitle}{\emph{40th IEEE Symposium on Security and Privacy
  (S\&P'19)}}.
\newblock


\bibitem[Kumar et~al\mbox{.}(2014)]%
        {CakeML}
\bibfield{author}{\bibinfo{person}{Ramana Kumar}, \bibinfo{person}{Magnus~O.
  Myreen}, \bibinfo{person}{Michael Norrish}, {and} \bibinfo{person}{Scott
  Owens}.} \bibinfo{year}{2014}\natexlab{}.
\newblock \showarticletitle{{CakeML}: A Verified Implementation of {ML}}. In
  \bibinfo{booktitle}{\emph{Proceedings of the 41st ACM SIGPLAN-SIGACT
  Symposium on Principles of Programming Languages}} (San Diego, California,
  USA) \emph{(\bibinfo{series}{POPL '14})}. \bibinfo{publisher}{Association for
  Computing Machinery}, \bibinfo{address}{New York, NY, USA},
  \bibinfo{pages}{179--191}.
\newblock
\showISBNx{9781450325448}
\urldef\tempurl%
\url{https://ts.data61.csiro.au/publications/nicta_full_text/7494.pdf}
\showURL{%
\tempurl}
\newblock
\shownote{10.1145/2535838.2535841}.


\bibitem[Lee et~al\mbox{.}(2018)]%
        {twinsem}
\bibfield{author}{\bibinfo{person}{Juneyoung Lee}, \bibinfo{person}{Chung-Kil
  Hur}, \bibinfo{person}{Ralf Jung}, \bibinfo{person}{Zhengyang Liu},
  \bibinfo{person}{John Regehr}, {and} \bibinfo{person}{Nuno~P. Lopes}.}
  \bibinfo{year}{2018}\natexlab{}.
\newblock \showarticletitle{Reconciling high-level optimizations and low-level
  code in LLVM}.
\newblock \bibinfo{journal}{\emph{Proc. ACM Program. Lang.}}
  \bibinfo{volume}{2}, \bibinfo{number}{OOPSLA}, Article
  \bibinfo{articleno}{125} (\bibinfo{date}{Oct.} \bibinfo{year}{2018}),
  \bibinfo{numpages}{28}~pages.
\newblock
\urldef\tempurl%
\url{https://doi.org/10.1145/3276495}
\showDOI{\tempurl}


\bibitem[Leroy(2009)]%
        {Leroy-backend}
\bibfield{author}{\bibinfo{person}{Xavier Leroy}.}
  \bibinfo{year}{2009}\natexlab{}.
\newblock \showarticletitle{A formally verified compiler back-end}.
\newblock \bibinfo{journal}{\emph{Journal of Automated Reasoning}}
  \bibinfo{volume}{43}, \bibinfo{number}{4} (\bibinfo{year}{2009}),
  \bibinfo{pages}{363--446}.
\newblock
\urldef\tempurl%
\url{http://xavierleroy.org/publi/compcert-backend.pdf}
\showURL{%
\tempurl}


\bibitem[Muller and Chong(2012)]%
        {refiners}
\bibfield{author}{\bibinfo{person}{Stefan Muller} {and}
  \bibinfo{person}{Stephen Chong}.} \bibinfo{year}{2012}\natexlab{}.
\newblock \showarticletitle{Towards a Practical Secure Concurrent Language}. In
  \bibinfo{booktitle}{\emph{Proceedings of the 25th Annual {ACM} {SIGPLAN}
  Conference on Object-Oriented Programming Languages, Systems, Languages, and
  Applications}}. \bibinfo{publisher}{ACM Press}, \bibinfo{address}{New York,
  NY, USA}, \bibinfo{pages}{57--74}.
\newblock


\bibitem[Protzenko et~al\mbox{.}(2020)]%
        {EverCrypt}
\bibfield{author}{\bibinfo{person}{Jonathan Protzenko}, \bibinfo{person}{Bryan
  Parno}, \bibinfo{person}{Aymeric Fromherz}, \bibinfo{person}{Chris
  Hawblitzel}, \bibinfo{person}{Marina Polubelova},
  \bibinfo{person}{Karthikeyan Bhargavan}, \bibinfo{person}{Benjamin
  Beurdouche}, \bibinfo{person}{Joonwon Choi}, \bibinfo{person}{Antoine
  Delignat-Lavaud}, \bibinfo{person}{Cédric Fournet}, \bibinfo{person}{Natalia
  Kulatova}, \bibinfo{person}{Tahina Ramananandro}, \bibinfo{person}{Aseem
  Rastogi}, \bibinfo{person}{Nikhil Swamy}, \bibinfo{person}{Christoph~M.
  Wintersteiger}, {and} \bibinfo{person}{Santiago Zanella-Beguelin}.}
  \bibinfo{year}{2020}\natexlab{}.
\newblock \showarticletitle{EverCrypt: A Fast, Verified, Cross-Platform
  Cryptographic Provider}. In \bibinfo{booktitle}{\emph{2020 IEEE Symposium on
  Security and Privacy (SP)}}. \bibinfo{pages}{983--1002}.
\newblock
\urldef\tempurl%
\url{https://doi.org/10.1109/SP40000.2020.00114}
\showDOI{\tempurl}


\bibitem[Schröder et~al\mbox{.}(2024)]%
        {DivideAndSurrender}
\bibfield{author}{\bibinfo{person}{Robin~Leander Schröder},
  \bibinfo{person}{Stefan Gast}, {and} \bibinfo{person}{Qian Guo}.}
  \bibinfo{year}{2024}\natexlab{}.
\newblock \bibinfo{title}{Divide and Surrender: Exploiting Variable Division
  Instruction Timing in {HQC} Key Recovery Attacks}.
\newblock \bibinfo{howpublished}{Cryptology ePrint Archive, Paper 2024/299}.
\newblock
\urldef\tempurl%
\url{https://eprint.iacr.org/2024/299}
\showURL{%
\tempurl}
\newblock
\shownote{\url{https://eprint.iacr.org/2024/299}}.


\bibitem[Sison and Murray(2019)]%
        {Sison_Murray_19}
\bibfield{author}{\bibinfo{person}{Robert Sison} {and} \bibinfo{person}{Toby
  Murray}.} \bibinfo{year}{2019}\natexlab{}.
\newblock \showarticletitle{Verifying That a Compiler Preserves Concurrent
  Value-Dependent Information-Flow Security}. In
  \bibinfo{booktitle}{\emph{International Conference on Interactive Theorem
  Proving}} (2019-9-6), Vol.~\bibinfo{volume}{141}. \bibinfo{publisher}{Schloss
  Dagstuhl}, \bibinfo{address}{Portland, USA}, \bibinfo{pages}{27:1--27:19}.
\newblock
\showISSN{18688969}
\urldef\tempurl%
\url{https://doi.org/10.4230/LIPIcs.ITP.2019.27}
\showDOI{\tempurl}


\bibitem[Swamy et~al\mbox{.}(2016)]%
        {Fstar}
\bibfield{author}{\bibinfo{person}{Nikhil Swamy},
  \bibinfo{person}{C\u{a}t\u{a}lin Hri\c{t}cu}, \bibinfo{person}{Chantal
  Keller}, \bibinfo{person}{Aseem Rastogi}, \bibinfo{person}{Antoine
  Delignat-Lavaud}, \bibinfo{person}{Simon Forest},
  \bibinfo{person}{Karthikeyan Bhargavan}, \bibinfo{person}{C\'{e}dric
  Fournet}, \bibinfo{person}{Pierre-Yves Strub}, \bibinfo{person}{Markulf
  Kohlweiss}, \bibinfo{person}{Jean-Karim Zinzindohoue}, {and}
  \bibinfo{person}{Santiago Zanella-B\'{e}guelin}.}
  \bibinfo{year}{2016}\natexlab{}.
\newblock \showarticletitle{Dependent types and multi-monadic effects in F*}.
  In \bibinfo{booktitle}{\emph{Proceedings of the 43rd Annual ACM
  SIGPLAN-SIGACT Symposium on Principles of Programming Languages}} (St.
  Petersburg, FL, USA) \emph{(\bibinfo{series}{POPL '16})}.
  \bibinfo{publisher}{Association for Computing Machinery},
  \bibinfo{address}{New York, NY, USA}, \bibinfo{pages}{256–270}.
\newblock
\showISBNx{9781450335492}
\urldef\tempurl%
\url{https://doi.org/10.1145/2837614.2837655}
\showDOI{\tempurl}


\bibitem[van~der Wall and Meyer(2025)]%
        {spec1}
\bibfield{author}{\bibinfo{person}{S\"{o}ren van~der Wall} {and}
  \bibinfo{person}{Roland Meyer}.} \bibinfo{year}{2025}\natexlab{}.
\newblock \showarticletitle{SNIP: Speculative Execution and Non-Interference
  Preservation for Compiler Transformations}.
\newblock \bibinfo{journal}{\emph{Proc. ACM Program. Lang.}}
  \bibinfo{volume}{9}, \bibinfo{number}{POPL}, Article \bibinfo{articleno}{51}
  (\bibinfo{date}{Jan.} \bibinfo{year}{2025}), \bibinfo{numpages}{30}~pages.
\newblock
\urldef\tempurl%
\url{https://doi.org/10.1145/3704887}
\showDOI{\tempurl}


\bibitem[Vassena et~al\mbox{.}(2021)]%
        {Blade}
\bibfield{author}{\bibinfo{person}{Marco Vassena}, \bibinfo{person}{Craig
  Disselkoen}, \bibinfo{person}{Klaus~von Gleissenthall},
  \bibinfo{person}{Sunjay Cauligi}, \bibinfo{person}{Rami~G\"{o}khan
  K\i{}c\i{}}, \bibinfo{person}{Ranjit Jhala}, \bibinfo{person}{Dean Tullsen},
  {and} \bibinfo{person}{Deian Stefan}.} \bibinfo{year}{2021}\natexlab{}.
\newblock \showarticletitle{Automatically eliminating speculative leaks from
  cryptographic code with {Blade}}.
\newblock \bibinfo{journal}{\emph{Proc. ACM Program. Lang.}}
  \bibinfo{volume}{5}, \bibinfo{number}{POPL}, Article \bibinfo{articleno}{49}
  (\bibinfo{date}{jan} \bibinfo{year}{2021}), \bibinfo{numpages}{30}~pages.
\newblock
\urldef\tempurl%
\url{https://doi.org/10.1145/3434330}
\showDOI{\tempurl}


\bibitem[Xia et~al\mbox{.}(2019)]%
        {itrees}
\bibfield{author}{\bibinfo{person}{Li-yao Xia}, \bibinfo{person}{Yannick
  Zakowski}, \bibinfo{person}{Paul He}, \bibinfo{person}{Chung-Kil Hur},
  \bibinfo{person}{Gregory Malecha}, \bibinfo{person}{Benjamin~C. Pierce},
  {and} \bibinfo{person}{Steve Zdancewic}.} \bibinfo{year}{2019}\natexlab{}.
\newblock \showarticletitle{Interaction trees: representing recursive and
  impure programs in Coq}.
\newblock \bibinfo{journal}{\emph{Proc. ACM Program. Lang.}}
  \bibinfo{volume}{4}, \bibinfo{number}{POPL}, Article \bibinfo{articleno}{51}
  (\bibinfo{date}{Dec.} \bibinfo{year}{2019}), \bibinfo{numpages}{32}~pages.
\newblock
\urldef\tempurl%
\url{https://doi.org/10.1145/3371119}
\showDOI{\tempurl}


\bibitem[Zinzindohou\'{e} et~al\mbox{.}(2017)]%
        {HACL}
\bibfield{author}{\bibinfo{person}{Jean-Karim Zinzindohou\'{e}},
  \bibinfo{person}{Karthikeyan Bhargavan}, \bibinfo{person}{Jonathan
  Protzenko}, {and} \bibinfo{person}{Benjamin Beurdouche}.}
  \bibinfo{year}{2017}\natexlab{}.
\newblock \showarticletitle{HACL*: A Verified Modern Cryptographic Library}. In
  \bibinfo{booktitle}{\emph{Proceedings of the 2017 ACM SIGSAC Conference on
  Computer and Communications Security}} (Dallas, Texas, USA)
  \emph{(\bibinfo{series}{CCS '17})}. \bibinfo{publisher}{Association for
  Computing Machinery}, \bibinfo{address}{New York, NY, USA},
  \bibinfo{pages}{1789–1806}.
\newblock
\showISBNx{9781450349468}
\urldef\tempurl%
\url{https://doi.org/10.1145/3133956.3134043}
\showDOI{\tempurl}


\end{thebibliography}

\end{document}